\documentclass[fleqn,11pt]{article}
\usepackage{amsthm,amsmath,amssymb}
\usepackage{geometry}
\usepackage{dsfont}
\usepackage{float}
\usepackage{natbib}
\usepackage{graphicx}
\usepackage{tikz}
\usepackage{pgfplots}
\pgfplotsset{compat=1.18}
\usepgfplotslibrary{groupplots}
\newcommand{\largestSizeDrawn}{30}

\definecolor{plannerGray}{gray}{0.55}
\definecolor{policyOrange}{RGB}{186,95,18}

\pgfplotsset{
  exampleTwoAxis/.style={
    width=\linewidth,
    height=0.34\textheight,
    axis lines*=left,
    tick align=outside,
    scaled y ticks=false,
    scaled x ticks=false,
    yticklabel style={/pgf/number format/fixed},
    label style={font=\small},
    tick label style={font=\small},
    legend style={font=\small, draw=none, fill=none, cells={anchor=west}},
    legend cell align=left,
  },
  plannerSettled/.style={plannerGray, line width=2.4pt, dashed,
    mark=none},
  plannerLine/.style={plannerGray, line width=2.4pt, mark=none},
  marketSettled/.style={black, line width=0.8pt, dashed, mark=none},
  marketLine/.style={black, line width=0.8pt, mark=none},
  solvedFeePoint/.style={policyOrange, only marks, mark=*, mark size=1.5pt},
}

\newlength{\figurenoteindent}
\AtBeginDocument{\setlength{\figurenoteindent}{\parindent}}
\newcommand{\figurenote}[1]{%
  \par\medskip
  \noindent\begin{minipage}{\linewidth}%
    \setlength{\parindent}{\figurenoteindent}%
    \footnotesize #1\par
  \end{minipage}}

\newcommand{\plannerSurplus}{0.150100213388556}
\newcommand{\laissezFaireSurplus}{0.0795715682512045}

\newcommand{\regulationSurplus}{0.0608544775508946}
\newcommand{\matureSizeMedian}{1.60103221856768}
\newcommand{\matureSizeSeventyFifth}{2.66583633148742}
\newcommand{\matureSizeNinetieth}{4.99306145305412}
\newcommand{\matureSizeNinetyFifth}{8.31381434585314}
\newcommand{\matureSizeNinetyEighth}{15.5716183502406}
\newcommand{\matureSizeNinetyNinth}{25.927889180934}
\newcommand{\gapOpensSize}{2.12758686885383}
\newcommand{\gapOpensPercentile}{66.6494790552953}
\newcommand{\bestEntryFee}{19}
\newcommand{\bestEntryFeeSurplus}{0.14778606511054}
\newcommand{\largestEntryFee}{60}
\newcommand{\lowestBeatingFee}{2}
\newcommand{\highestBeatingFee}{60}
\newcommand{\marketsSimulated}{1000000}
\newcommand{\poolWorkers}{18}
\newcommand{\equilibriaTimed}{61}
\newcommand{\secondsPerEquilibrium}{10.2045936099727}
\newcommand{\secondsPerValuation}{78.3402899002732}
\newcommand{\projectedDays}{0.656159568542482}

\newcommand{\annualCostOfCapital}{20}
\newcommand{\grossAnnualFactor}{1.2}
\newcommand{\discountFactor}{0.984921375260889}
\newcommand{\developmentMonths}{36}
\newcommand{\simulationMonths}{1200}
\newcommand{\demandIntercept}{1}
\newcommand{\growthStep}{1.04}
\newcommand{\growthPercent}{4}
\newcommand{\growthProbability}{0.5}
\newcommand{\marketSizes}{240}
\newcommand{\technologyTypes}{30}
\newcommand{\costDecline}{0.04}
\newcommand{\advanceProbability}{0.5}
\newcommand{\entrantMarginalCost}{1}
\newcommand{\frontierMarginalCost}{0.313486180882605}
\newcommand{\medianContinuationCost}{1}
\newcommand{\lowestContinuationCost}{0.7937005259841}
\newcommand{\exitElasticity}{3}
\newcommand{\monthsAfterMaturity}{36}

\newcommand{\valueToleranceExponent}{-40}
\newcommand{\discountAtHorizon}{1.20746734724139e-08}
\newcommand{\truncationBound}{1.93869111622051e-05}
\newcommand{\lernerTarget}{0.4}
\newcommand{\demandSlope}{0.154651051841072}
\newcommand{\costCurvature}{0.0389081433841241}

\newcommand{\solvingProcessor}{Apple M5 Pro}
\newcommand{\solvingSockets}{1}
\newcommand{\solvingCoresPerSocket}{18}

\newcommand{\largestAccumulatedStandardError}{0.000259469207805735}
\newcommand{\largestFirmCountStandardError}{0.00895629169129435}
\newcommand{\largestLeaderCostStandardError}{0.000224779676832515}

\usepackage{siunitx}
\newcommand{\surplus}[1]{%
  \num[round-mode = places, round-precision = 4]{#1}}

\newcommand{\parameter}[1]{%
  \num[round-mode = figures, round-precision = 3]{#1}}

\newcommand{\chosen}[1]{\num{#1}}

\newcommand{\marketsize}[1]{%
  \num[round-mode = figures, round-precision = 4]{#1}}

\newcommand{\scientific}[1]{%
  \num[round-mode = figures, round-precision = 3,
       exponent-mode = scientific]{#1}}

\newcommand{\counted}[1]{\num{#1}}

\newcommand{\probability}[1]{%
  \edef\probabilityReciprocal{\fpeval{1/(#1)}}%
  \ifnum\fpeval{round(\probabilityReciprocal) == \probabilityReciprocal}=1
    \ensuremath{1/\probabilityReciprocal}%
  \else\parameter{#1}%
  \fi}

\newcommand{\seconds}[1]{%
  \num[round-mode = figures, round-precision = 3]{#1}}

\newcommand{\powerOfTwo}[1]{\ensuremath{2^{#1}}}

\newcommand{\surplusGap}{\fpeval{\plannerSurplus - \laissezFaireSurplus}}

\usetikzlibrary{positioning,calc,shapes,shapes.multipart,shapes.misc,arrows,chains,fit,plotmarks}
\def\stackedpayoffs#1#2{%
  \begin{array}{c}\\ #1\hspace{1cm}\\[0.5cm]\hspace{1cm}#2\vspace{0.1cm} \end{array}
}

\usepackage[pagebackref=false]{hyperref}
\hypersetup{colorlinks=true,linkcolor=blue,urlcolor=blue,citecolor=red,
     pdftitle={Entry, Exit, and Technological Progress in Markov-Perfect Duopoly},
    pdfauthor=Jaap H. Abbring and Jeffrey R. Campbell and Nan Yang,
    pdfsubject=Dynamic Oligopoly Theory,
  pdfkeywords={sunk costs, barriers to entry, equilibrium uniqueness, Fake Backward Induction, competition policy},
    pdfdisplaydoctitle=true,
    bookmarksopen=true}

\theoremstyle{plain}
\newtheorem{theorem}{Theorem}
\newtheorem{lemma}{Lemma}
\newtheorem{definition}{Definition}
\newtheorem{alg}{Algorithm}
\theoremstyle{definition}
\newtheorem{assumption}{Assumption}
\newcommand{\cites}[1]{\citeauthor{#1}'s (\citeyear{#1})}
\newcommand{\true}[0]{\mathds{1}}
\newfloat{procedure}{p}{prc}
\floatname{procedure}{Procedure}

\pgfdeclaremask{atari-mask}{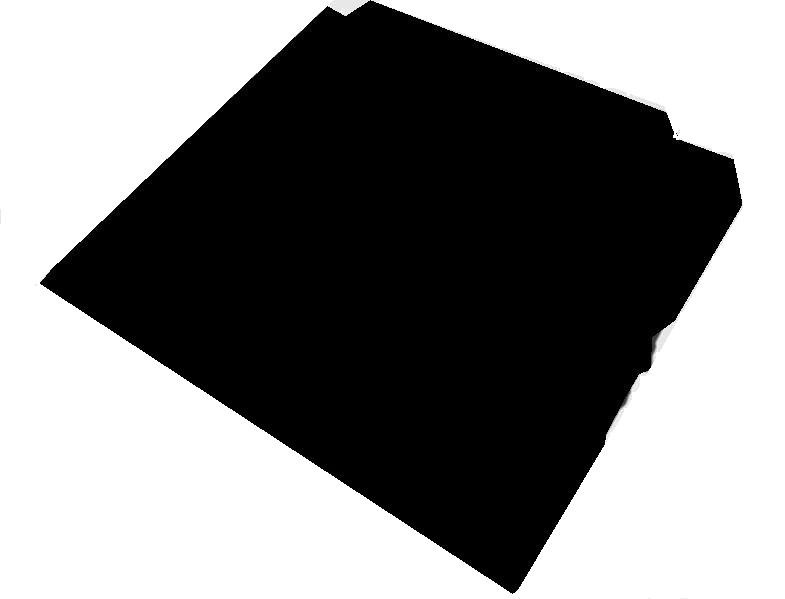}
\pgfdeclareimage[width=15mm,mask=atari-mask,interpolate=true]{atari}{atari400}

\title{\vspace*{-1em}
Entry, Exit, and Technological Progress in Markov-Perfect Duopoly\thanks{We are grateful to Luis Cabral, Bas van Heiningen, and participants at the 2022 Dynamic Structural Econometrics conference at the Australian National University for their helpful comments. The research of Jaap Abbring was financially supported by the Netherlands Organisation for Scientific Research (NWO) through Vici grant 453-11-002. Nan Yang gratefully acknowledges the financial support from the National University of Singapore Start-up Grant R-313-000-100-133. This paper's conclusions are those of the authors. They do not necessarily reflect the opinions of the Federal Reserve Bank of Chicago, the Federal Reserve System, or its Board of Governors.
} 
} 
\author{%
Jaap H. Abbring\thanks{%
Department of Econometrics \& OR, Tilburg University; and CEPR.
E-mail: \href{mailto:jaap@abbring.org}{jaap@abbring.org}}
\and Jeffrey R. Campbell\thanks{%
Department of Economics, University of Notre Dame; Department of Econometrics \& OR, Tilburg University; Federal Reserve Bank of Chicago; and CEPR. E-mail:
\href{mailto:jcampb24@nd.edu}{jcampb24@nd.edu}}
\and Nan Yang\thanks{%
E-mail:
\href{mailto:yangnanpku@gmail.com}{yangnanpku@gmail.com}
}}

\begin{document}
\setcounter{page}{0}

 \maketitle

\begin{abstract}
This paper establishes the existence and essential uniqueness of an intuitively refined Markov-perfect equilibrium in a duopoly game with sunk costs, barriers to entry, firm-specific technological progress, and irreversible exit.  We develop an algorithm that quickly computes the equilibrium by solving for fixed points of a finite sequence of auxiliary contraction mappings, which we label \emph{Fake Backward Induction}. We illustrate the game's practical value with a quantitative analysis of competition policy in a stochastically growing market with firms idiosyncratically learning by aging. Computing one equilibrium of this example takes \seconds{\secondsPerEquilibrium} seconds on one \solvingProcessor\ core.
\end{abstract}

\thispagestyle{empty}

\newpage \setcounter{page}{1}


\section{Introduction}

This paper studies the Markov-perfect equilibria of a dynamic duopoly game of entry and exit. The game's potential entrants pay sunk costs to become active, and its active firms' technologies idiosyncratically improve over time. Exogenously determined market conditions follow a Markov process. A firm's static payoff rises in its own technology and declines with the entry and technological development of a rival. Firms can exit irreversibly to avoid losses. We focus on symmetric equilibria in which a firm with superior technology never exits while leaving behind a less advanced rival \citep{jet93:cabral}, which we call \emph{natural} Markov-perfect equilibria. We show that such an equilibrium exists and that all natural equilibria imply the same joint distribution for firms' technologies and market conditions. We develop an algorithm that computes one such equilibrium quickly by finding the fixed points of a finite sequence of low-dimensional contraction mappings.  

We leverage one insight into a symmetric Markov-perfect equilibrium's payoffs: when two firms nontrivially mix between exit and survival, the expected payoff from survival equals zero. This allows calculation of firms' expected continuation values at some nodes of the game tree without knowing how the game will proceed. In turn, this facilitates using a particular induction to determine firms' payoffs with standard dynamic programming methods applied to \emph{simplified} specifications of firms' strategic environments.

Our algorithm builds on applications of backward induction to compute the equilibria of dynamic directional games \citep*[e.g.,][]{ecta94:cabralriordan,ier12:juddetal}. \citet*{iskhakov2016recursive} systematize these procedures into an algorithm that computes all of these games' equilibria. In the games they consider, primitive restrictions on state transitions partially order the state space: state B comes after state A if the game can move from A to B but not from B to A. Iterating backward through this partial order computes equilibrium outcomes and continuation values recursively. Our game lacks such directionality, but nodes at which identical duopolists exit with positive probability act like terminal nodes. For this reason, we say that our algorithm uses \emph{Fake Backward Induction} (henceforth ``FBI'').

\cites{res95:ericsonpakes} dynamic oligopoly framework is more general than ours: it allows firms to control their rate of technological progress with investment and allows three or more firms to be active. \citet*{elsevier07:doraszelskipakes} review the large literature that builds on this framework. The game's state may include any combination of the firms' payoff-relevant characteristics, so the number of possible states grows exponentially with the number of firms. \citet*{ecta01:pakesmcguire} develop a stochastic algorithm to overcome the resulting curse of dimensionality. Nevertheless, we know little about the framework's equilibrium set beyond computational demonstrations of multiplicity \citep*[e.g.,][]{ecta10:besankoetal}.\footnote{Equilibrium multiplicity need not impede the empirical application of similar frameworks. For example, \citet*{ecta07:bajarietal} and \citet*{rje13:dunneetal} assume that firms in markets with identical primitives play the \emph{same} equilibrium. However, equilibrium multiplicity might impede counterfactual policy analysis. \citet{nwu23:reguantpareschi} bound the counterfactual outcomes that any equilibrium could support.}

Our results show that equilibrium multiplicity disappears from duopoly markets when firms choose continuation rather than investment intensity and attention is restricted to natural Markov-perfect equilibria. Our framework might therefore provide a useful duopoly benchmark for industries in which the most consequential investment decisions concern the continuation of a fixed product-development program, such as pharmaceuticals \citep{ubc23:khmelnitskaya}. We guarantee our algorithm's convergence. Therefore, its application can reliably produce starting points for homotopy \citep*{or10:borkovskyetal}, which in turn can trace the equilibrium sets of richer specifications. In this way, we hope to follow \citet*{rand94:pakesmcguire} by providing a tool to make Markov-perfect equilibrium analysis easier.\footnote{The analysis of \emph{oblivious equilibrium} \citep*{ecta08:weintraubetal} provides another approach to reducing the difficulty of Markov-perfect equilibrium analysis. That characterizes industry dynamics involving many competitors, whereas our game characterizes a market with at most two firms.} 

The remainder of the paper proceeds as follows. Section \ref{sec:duopolyGame} presents the duopoly game's primitives and defines ``natural Markov-perfect equilibrium.'' Section \ref{sec:twoTypes} develops intuition by constructing such an equilibrium in the special case with two technology types using Fake Backward Induction. Section \ref{sec:constructionUniqueness}  generalizes this construction to an arbitrary number of technology types, and establishes that the natural Markov-perfect equilibrium it computes is essentially unique. Section \ref{sec:competitionPolicy} applies the algorithm to a quantitative analysis of regulatory schemes that change firms' entry and exit incentives. Section \ref{sec:conclusion} provides concluding remarks on the algorithm's computational speed. The Appendix proves the lemmas of Section \ref{sec:twoTypes}. Our Online Supplement \citep*{duopolysupp:abbringetal} proves Theorems \ref{th:existence} and \ref{th:uniqueness} from Section \ref{sec:constructionUniqueness}.


\section{The Duopoly Game\label{sec:duopolyGame}}

The duopoly game characterizes a market in discrete time, indexed with $t$. At any given time, either zero, one, or two firms serve the market. Each firm active in period $t$ has an idiosyncratic technology type $K_t$ that takes values in $\mathbb{K}\equiv\{1,\ldots,\check{k}\}$.\footnote{Although we will consistently refer to $K_t$ as reflecting the firm's technology type, it could have a demand-side interpretation.}  We use capital Roman letters to denote random variables and reserve small Roman letters for their realizations. Here, $\check{k}<\infty$ is the  highest possible technology type. To distinguish an active firm's technology type from its competitor's, we denote the latter with $X_t$. We set $X_t=0$ when there is no active competitor. Firms' period $t$ payoffs depend on their types and on an \emph{exogenous state} $Y_t$, which may embody both demand and supply factors. We denote its support with $\mathcal{Y}$.

The sequence of events and actions within period $t$ proceeds as follows.

\begin{enumerate}
\item The period $t$ technology types of all firms surviving from period $t-1$ and a realization $y_t$ of the exogenous state $Y_t$ are drawn from independent Markov processes. 
\item Active firms receive profits from serving the market. For a type $k$ firm facing a type $x$ rival, these equal $\pi(k,x,y_t)\geq 0$. 
\item If fewer than two firms are active, entry can take place. If the market has a single incumbent, one potential entrant decides on entry; if it is empty, two potential entrants make their choices sequentially.\footnote{A firm benefits from being first in this sequence; but since the firms are identical after entry, all assignments of firms to positions in this queue are observationally equivalent modulo firms' names.} A firm that stays out of the market earns a payoff of zero. Such a refusal to enter also ends the period's entry stage, so a second potential entrant loses its opportunity whenever its potential rival earlier in the queue remains inactive \citep*[as in][]{vsmpidTheory18:abbringetal}.  An entrant into an empty market pays a sunk cost $\varphi(0,y_t)\geq 0$ and immediately gets the technology type $1$.  An entrant into a market with a type $x$ incumbent or earlier entrant pays a sunk cost $\varphi(x,y_t)$ and joins the set of active firms with technology type $1$.\footnote{When two firms enter an initially empty market, the first entrant pays $\varphi(0,y_t)$ and the second pays $\varphi(1,y_t)$. } The potential entrants for period $t$ do \emph{not} have the option to delay their entry opportunities. Instead, they are replaced in the next period by new firms. 
\item All active firms -- incumbents and new entrants -- simultaneously choose between continuation and exit. Exit is irreversible, and its payoff is zero. Survival requires paying the fixed cost of continuation $\kappa(y_t)\geq 0$. The surviving firms move to period $t+1$.
\end{enumerate}
Firms' entry and exit decisions maximize their expected profit streams discounted with the factor $\beta\in[0,1)$. Each firm can randomize with a private device that generates any number of draws from the uniform distribution on $[0,1]$. The draws are independent of each other and of all payoff-relevant variables. Only the firm itself observes its draws, so firms cannot use them to coordinate their actions.

We impose some restrictions on how a firm's payoffs depend on its own and its competitor's types.
\begin{assumption}[Monotone and Bounded Profits] 
\label{ass:monotoneBoundedProfits} 
\;
\begin{enumerate}
\item \label{ass:monotoneBoundedProfits:ownType} For all $k\in\{1,\ldots,\check{k}-1\}$, all $x\in\{0\}\cup\mathbb{K}$, and all $y\in\mathcal{Y}$, $\pi(k,x,y)\leq\pi(k+1,x,y)$;
\item \label{ass:monotoneBoundedProfits:rivalType} for all $(k,x)\in\mathbb{K}\times\mathbb{K}$ and all $y\in\mathcal{Y}$, $\pi(k,x,y)\leq\pi\left(k,x-1,y\right)$; and
\item \label{ass:monotoneBoundedProfits:bounded} there is a $\check{\pi}<\infty$ such that for all $(k,x)\in\mathbb{K}\times\left(\{0\}\cup\mathbb{K}\right)$ and $y\in\mathcal{Y}$, $\pi(k,x,y)\leq \check{\pi}$.
\end{enumerate}
\end{assumption}

\begin{assumption}[Barriers to Entry]
\label{ass:entryBarriers} For all $y\in\mathcal{Y}$ and $x\in\mathbb{K}$, $0\leq\varphi(x-1,y)\leq \varphi(x,y)$.
\end{assumption}

\noindent Assumptions \ref{ass:monotoneBoundedProfits}.\ref{ass:monotoneBoundedProfits:ownType} and \ref{ass:monotoneBoundedProfits}.\ref{ass:monotoneBoundedProfits:rivalType} require that a firm's profits increase with its own technology and decrease with the addition of a rival or the improvement of the rival's technology. Because these orderings must hold in every state $y$, Assumption \ref{ass:monotoneBoundedProfits} excludes the possibility that firms' profitability ranking varies with demand.\footnote{\cite{rand85:ghemawatnalebuff} exemplify that possibility. In their model, a small monopolist remains profitable at demand levels too low to support a larger one, so the largest firm exits first from a declining market.} Assumption \ref{ass:entryBarriers} ensures that a firm's entry cost weakly increases with a rival's presence and technology type. There are \emph{economic barriers to entry} \citep*{aer2004McAfeeMialonWilliams} if firms pay \emph{strictly} more to enter a market with a type $x\geq 1$ incumbent than to enter an empty market, $\varphi(x,y)>\varphi(0,y)$. Assumption \ref{ass:monotoneBoundedProfits}.\ref{ass:monotoneBoundedProfits:bounded} is a technical condition that ensures that expected discounted profits are bounded.

We also impose some structure on the stochastic processes for the firms' technology types and the exogenous state. Active firms' technology types follow independent Markov chains over $\mathbb{K}$ with transition probability matrix $\Pi$. The $(i,j)$-th element of $\Pi$, $\Pi_{ij}$, is the probability that a type $i$ firm progresses to type $j$. We make two assumptions on $\Pi$ that align Assumption \ref{ass:monotoneBoundedProfits}.\ref{ass:monotoneBoundedProfits:ownType}'s static notion of having better technology with its dynamic consequences. 
\begin{assumption}[Monotone Technology Dynamics]
\label{ass:monotoneTechnology}
For all $k\in\{1,\ldots,\check{k}-1\}$ and all $k^\prime\in\mathbb{K}$,
$\sum_{j=k'}^{\check{k}}\Pi_{kj}\leq\sum_{j=k'}^{\check{k}}\Pi_{(k+1)j}$.
\end{assumption}
\begin{assumption}[No Technology Regress]
$\Pi$ is upper triangular.
\label{ass:noRegress}
\end{assumption}

\noindent Assumption \ref{ass:monotoneTechnology} requires that the next period's type for a firm with a higher current type stochastically dominates the next period's type for a firm with a lower current type: $\Pr\left[K_{t+1}\geq k'\;\vline\;K_{t}=k\right]\leq \Pr\left[K_{t+1}\geq k'\;\vline\;K_{t}=k+1\right]$. Assumption \ref{ass:noRegress} ensures that firms' technology types never regress. 

Our final three assumptions impose technical restrictions on the Markov process for the exogenous state. 
\begin{assumption}[Conditional Independence]
\label{ass:conditionalIndependence}
The exogenous state $Y_t=(C_t,Z_t)$, where the \emph{persistent state} $C_t$ follows a $\mathcal{C}$-valued first-order Markov process and the \emph{shock} $Z_t$ is identically and independently distributed over time on $\mathcal{Z}$. The process $\{Z_t\}$ is independent of $\{C_t\}$, and $\mathcal{Y}=\mathcal{C}\times\mathcal{Z}$. 
\end{assumption}
\begin{assumption}[Continuously Distributed Shocks]
\label{ass:continuousShocks}
For all $c\in\mathcal{C}$, $x\in\left(\{0\}\cup\mathbb{K}\right)$, and  $u\in\mathbb{R}$; $\Pr\left[\kappa(c,Z_t)=u\right]=\Pr\left[\kappa(c,Z_t)+\varphi(x,(c,Z_t))=u\right]=0$.
\end{assumption}
\begin{assumption}[Finite Persistent State]
\label{ass:finitePersistentState}
The support $\mathcal{C}$ of $C_t$ is a finite set.
\end{assumption}
\noindent Assumption \ref{ass:conditionalIndependence} implies that firms contemplating entry or exit do not use $Z_t$ to predict future payoffs. This is a version of \cites{ecta87:rust} conditional independence assumption, which underlies much of the econometrics of dynamic discrete choices and games.\footnote{\cite{rand99:gowrisankaran} introduces privately observed shocks to firms' scrap values into the \cite{res95:ericsonpakes} framework, and \cite{rand10:doraszelskisatterthwaite} use such shocks to firms' scrap values and set-up costs to implement Harsanyi purification. Those privately observed shocks could also serve as econometric errors. In contrast, we focus on mixed strategies in a game of complete information about payoff-relevant primitives and actions. } Assumptions \ref{ass:conditionalIndependence} and \ref{ass:continuousShocks} ensure that potential entrants and incumbent monopolists contemplating continuation are indifferent between their options with probability zero.  Assumption \ref{ass:finitePersistentState} allows us to center our equilibrium calculations on post-survival value functions with finite domains, so that they can be directly implemented on a digital computer.

We next consider firms' choices in Markov-perfect equilibria. In these, firms' strategies depend only on payoff-relevant variables. For a firm contemplating entry,  the payoff-relevant variables are the technology type of any incumbent rival and the current value of $Y_t$. If a firm is instead making its survival decision, the payoff-relevant variables are these and the firm's own technology type.

A \emph{Markov strategy} for a firm is a pair of functions $a_E:\left(\{0\}\cup\mathbb{K}\right)\times\mathcal{Y}\rightarrow \{0,1\}$ and $a_S:\mathbb{K}\times\left(\{0\}\cup\mathbb{K}\right)\times\mathcal{Y}\rightarrow[0,1].$ The \emph{entry rule} $a_E$ assigns an indicator for becoming active given an entry opportunity to each possible value $(x,y)$ of a potential entrant's payoff-relevant variables, where $x$ refers to either an incumbent's or an earlier entrant's type. We take the entry rule to be pure because potential entrants' decisions proceed sequentially. Similarly, the \emph{survival rule} $a_S$ gives a probability of being active in the next period given that the firm is currently active for each possible value $(k,x,y)$ of its payoff-relevant variables. Firms may use their private devices to implement survival probabilities in $(0,1)$.

Nothing payoff relevant distinguishes the game's firms from each other ex ante: they select their actions from identical action sets, and the same functions $\pi$, $\varphi$, and $\kappa$ determine their payoffs.  If such identical firms follow distinct strategies, then their choices depend on the only characteristic distinguishing them, their permanent but payoff-irrelevant names. \citet{jet01:maskintirole} require Markov strategies to exclude payoff-irrelevant variables from influencing play, and so allowing firms to follow different strategies violates that requirement's spirit if not its letter. For that reason, we restrict attention to \emph{symmetric} equilibria. In these, all firms follow the same strategy.

Firms' expected discounted profits in two of each period's nodes are of particular interest, the \emph{post-entry} value and the \emph{post-survival} value. The post-entry value $v_E(k,x,y)$ equals the expected discounted profits just after all entry decisions have been sequentially realized. When evaluated at $k=1$, it gives the payoffs to a potential producer from entering.  The post-survival value $v_S(k,x,c)$ equals the expected discounted profits just after all survival decisions have been realized. It depends only on the exogenous state's persistent component, because the conditionally independent shock's value neither directly affects future payoffs nor helps predict them. Here and throughout, $c$ denotes the persistent component of the exogenous state's value $y$. In $v_E$ and $v_S$, $k$ and $x$ are the firm's and its rival's types at that node, with $x=0$ when the firm has no rival. 

When all firms follow the Markov strategy $(a_E,a_S)$, the post-entry and post-survival values satisfy 
\begin{align}
\label{eq:vE}
&v_E(k,x,y)=\\
&\nonumber\left\{\begin{array}{ll} 
a_S(k,0,y)\left(v_S(k,0,c)-\kappa(y)\right)&\text{if }x=0,\\
a_S(k,x,y)\Bigg(a_S(x,k,y)v_S(k,x,c)+\left(1-a_S(x,k,y)\right)v_S(k,0,c)-\kappa(y)\Bigg)&\text{if }x>0;
\end{array}\right.
\end{align}
and   
\begin{align}
\label{eq:vS}
&v_S(k,x,c)=\\
&\nonumber\left\{\begin{array}{ll}
 \beta\mathbb{E}\Bigg[\pi(K^\prime,0,Y^\prime) + a_E(K^\prime,Y^\prime)v_E(K^\prime,1,Y^\prime)&\\ 
 \hspace*{9em}+(1-a_E(K^\prime,Y^\prime))v_E(K^\prime,0,Y^\prime)\;\vline\;C=c,K=k\Bigg]
&\text{if }x=0,\\
 \beta\mathbb{E}\Bigg[\pi(K^\prime,X^\prime,Y^\prime) + v_E(K^\prime,X^\prime,Y^\prime)\;\vline\;C=c,K=k,X=x\Bigg]&\text{if }x>0.
\end{array}\right.
\end{align}

\noindent Here, we denote next period's value of a generic random variable $M$ with $M^\prime$. 

If a strategy $(a_E,a_S)$ forms a Markov-perfect equilibrium, then no firm can profit by deviating from it once. Therefore
\begin{align}
\label{eq:aE}
&a_E(x,y) \in\\
&\nonumber\left\{\begin{array}{ll} 
\arg\underset{a\in\{0,1\}}{\max}\;a\Bigg(a_E(1,y)v_E(1,1,y)+\left(1-a_E(1,y)\right)v_E(1,0,y)-\varphi(0,y)\Bigg)&\text{if }x=0,\\
\arg\underset{a\in\{0,1\}}{\max}\;a\left(v_E(1,x,y)-\varphi(x,y)\right)&\text{if }x>0;
\end{array}\right.
\end{align}
and
\begin{align}
\label{eq:aS}
&a_S(k,x,y) \in\\
&\nonumber\left\{\begin{array}{ll} 
\arg\underset{a\in[0,1]}{\max}\;a\left(v_S(k,0,c)-\kappa(y)\right)&\text{if }x=0,\\
\arg\underset{a\in[0,1]}{\max}\;a\Bigg(a_S(x,k,y)v_S(k,x,c)+\left(1-a_S(x,k,y)\right)v_S(k,0,c)-\kappa(y)\Bigg)&\text{if }x>0.
\end{array}\right.
\end{align}

\noindent \citet*[][Theorem S1 of the Supplemental Material]{vsmpidTheory18:abbringetal} establish that the converse is also true.\footnote{\label{fn:oneShotDeviation}Standard statements of this ``one-shot deviation principle,'' such as \cites{mit91:fudenbergtirole} Theorem 4.2, require payoffs to be ``continuous at infinity.'' Our game's payoffs might fail this condition because $\kappa(y)$ may be unbounded from above. The Supplemental Material is available with the article at \href{https://doi.org/10.3982/ECTA14060}{doi.org/10.3982/ECTA14060}.} Consequently, a strategy $(a_E,a_S)$ forms a Markov-perfect equilibrium if and only if \eqref{eq:vE}--\eqref{eq:aS} hold.

Whenever two firms with different technology types would both earn losses from continuing together, yet each would profit from surviving alone, a Markov-perfect strategy could dictate that the firm with the worse technology exits for sure, that its rival with the better technology exits for sure, or that both firms play mixed strategies. In the last two possibilities, the firm with the higher monopoly payoff exits more often than its rival. Two considerations lead us to follow \cite{jet93:cabral} by selecting equilibria in which a high-type firm never exits when a low-type competitor survives. First, survival increases with profitability in the data \citep*{aer08:fosteretal}. Second, when a firm with higher profits also adds more to social surplus, a social planner who makes every entry and continuation decision never has a better firm exit while a worse one continues. An equilibrium that did so would qualitatively depart from the planner's welfare benchmark. We refer to the selected equilibria as natural.\footnote{\cite{ecta10:abbringcampbell} use a similar restriction, that an older firm never exits while leaving behind a younger rival, to structure their equilibrium analysis and computation.}
\begin{definition}
\label{def:naturalEquilibrium} A \emph{natural} Markov-perfect equilibrium is a symmetric Markov-perfect equilibrium $(a_E,a_S)$ such
that for all $(h,l)\in\mathbb{K}^2$ with $l<h$ and all $y\in \mathcal{Y}$; $a_S(l,h,y)>0$ implies that $a_S(h,l,y)=1.$
\end{definition}

Henceforth, ``equilibrium'' refers to a natural Markov-perfect equilibrium.


\section{Fake Backward Induction with Two Technology Types\label{sec:twoTypes}}

To lay the groundwork for the game's analysis, we first construct an equilibrium for the special case with $\check{k}=2$ and strictly positive entry costs $\varphi$.\footnote{Section \ref{sec:constructionUniqueness}'s algorithmic-functionality and equilibrium-existence result for the game with general $\check{k}$ (Theorem \ref{th:existence}) allows zero entry costs (as in Assumption \ref{ass:entryBarriers}). We do, however, exclude zero entry costs to establish essential uniqueness (Theorem \ref{th:uniqueness}).} This displays the moving parts of FBI without undue complications. To emphasize that we are constructing a \emph{candidate} equilibrium, we denote the calculated payoffs and strategies with $w_S(k,x,c)$, $w_E(k,x,y)$, $\alpha_S(k,x,y)$, and $\alpha_E(k,y)$.  As above, we use $v_S(k,x,c)$, $v_E(k,x,y)$, $a_S(k,x,y)$, and $a_E(k,y)$ to denote equilibrium payoffs and strategies.

This procedure consists of six steps. The first five take us through a partition of the state space. In each of these, we compute the equilibrium payoffs in the states considered by finding the fixed point of a contraction mapping. These mappings impose necessary conditions on equilibrium payoffs, so the computed values equal the only possible equilibrium payoffs. The results from the completed steps are used as inputs to the following steps.

Figure \ref{fig:twoTypesFakeBackwardInduction} illustrates this sequence of computations. Each of its six panels contains five rectangles divided into top and bottom halves. A rectangle represents a market structure, with the numbers in its top and bottom halves giving the number of active type $1$ and type $2$ firms, respectively. Labels mark the six steps, and the arrows between them indicate their sequence. If the value function for a firm with a given type in a particular market structure is being calculated in a step, then its half-rectangle is colored red. In subsequent steps it is colored blue to indicate that the corresponding value function is in hand.


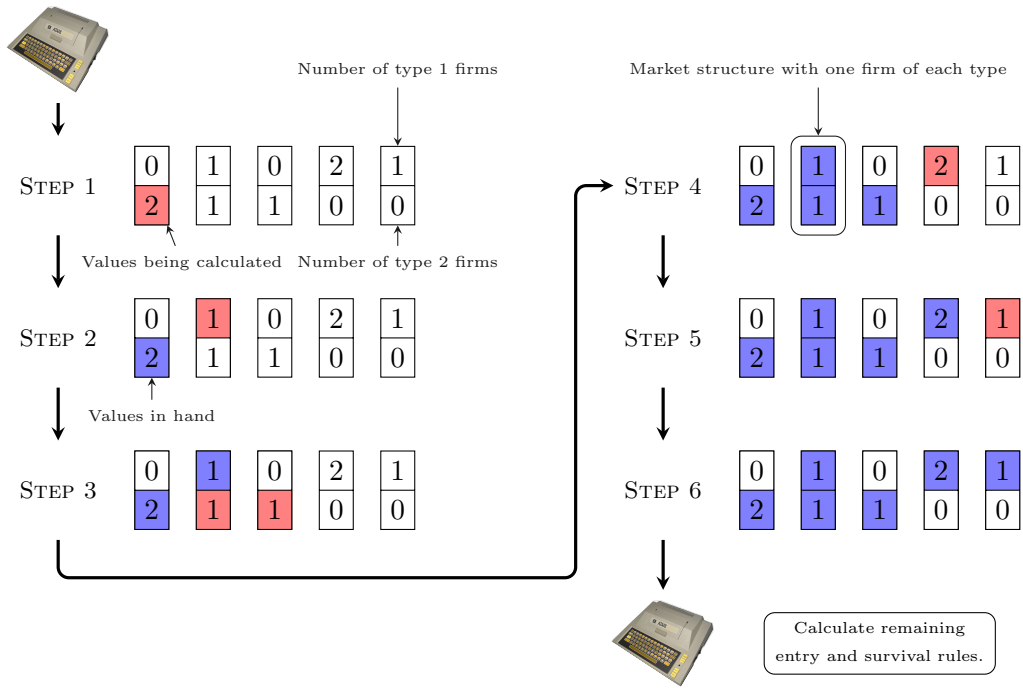
\begin{figure}[tp]
\caption{\label{fig:twoTypesFakeBackwardInduction}Fake Backward Induction with Two Technology Types}
\begin{center}

\begin{tikzpicture}[node distance=3.5mm,
    unknown/.style={rectangle split, draw, rectangle split parts=2,rectangle split
part align={center,center},minimum width=.1cm},
    computing bottom/.style = {rectangle split, draw, rectangle split parts=2,rectangle split
part align={center,center},minimum width=.1cm ,rectangle split part fill = {white,red!50}},
    known bottom/.style ={rectangle split, draw, rectangle split parts=2,rectangle split
part align={center,center},minimum width=.1cm ,rectangle split part fill = {white,blue!50}},
    computing top/.style={rectangle split, draw, rectangle split parts=2,rectangle split
part align={center,center},minimum width=.1cm ,rectangle split part fill = {red!50,white}},
    known top/.style ={rectangle split, draw, rectangle split parts=2,rectangle split
part align={center,center},minimum width=.1cm ,rectangle split part fill = {blue!50,white}},
    computing bottom known top/.style={rectangle split, draw, rectangle split parts=2,rectangle split
part align={center,center},minimum width=.1cm ,rectangle split part fill = {blue!50,red!50}},
    known/.style={rectangle split, draw, rectangle split parts=2,rectangle split
part align={center,center},minimum width=.1cm ,rectangle split part fill = {blue!50,blue!50}},
    step label/.style={inner sep=0.5pt, minimum size=13mm, font=\footnotesize\scshape},
]

\node at (-3.5,5.8)  (start here) {\pgfuseimage{atari}};


\node[step label] at (-3.5,4)  (computer1) {Step 1};

\node[unknown,right=of computer1](step 1 market structure 0 2){0\nodepart{second}2};

\node[unknown,right=of step 1 market structure 0 2] (step 1 market structure 1 1) {1\nodepart{second}1};
\node[unknown,right=of step 1 market structure 1 1] (step 1 market structure 0 1) {0\nodepart{second}1};
\node[unknown,right=of step 1 market structure 0 1] (step 1 market structure 2 0) {2\nodepart{second}0};
\node[unknown,right=of step 1 market structure 2 0] (step 1 market structure 1 0) {1\nodepart{second}0};

\node[computing bottom] at (step 1 market structure 0 2)
{0\nodepart{second}2};


\node[step label,below=7mm of computer1] (computer2){Step 2};

\node[unknown,right=of computer2](step 2 market structure 0 2){0\nodepart{second}2};

\node[unknown,right=of step 2 market structure 0 2] (step 2 market structure 1 1) {1\nodepart{second}1};
\node[unknown,right=of step 2 market structure 1 1] (step 2 market structure 0 1) {0\nodepart{second}1};
\node[unknown,right=of step 2 market structure 0 1] (step 2 market structure 2 0) {2\nodepart{second}0};
\node[unknown,right=of step 2 market structure 2 0] (step 2 market structure 1 0) {1\nodepart{second}0};

\node[known bottom] at (step 2 market structure 0 2) {0\nodepart{second}2};

\node[computing top] at (step 2 market structure 1 1) {1\nodepart{second}1};

\node[step label,below=7mm of computer2] (computer3) {Step 3};

\node[unknown,right=of computer3](step 3 market structure 0 2){0\nodepart{second}2};

\node[unknown,right=of step 3 market structure 0 2] (step 3 market structure 1 1) {1\nodepart{second}1};
\node[unknown,right=of step 3 market structure 1 1] (step 3 market structure 0 1) {0\nodepart{second}1};
\node[unknown,right=of step 3 market structure 0 1] (step 3 market structure 2 0) {2\nodepart{second}0};
\node[unknown,right=of step 3 market structure 2 0] (step 3 market structure 1 0) {1\nodepart{second}0};

\node[known bottom] at (step 3 market structure 0 2) {0\nodepart{second}2};

\node[computing bottom known top] at (step 3 market structure 1 1) {1\nodepart{second}1};

\node[computing bottom] at (step 3 market structure 0 1) {0\nodepart{second}1};

\node[step label] at (4.5,4) (computer4){Step 4};

\node[unknown,right=of computer4](step 4 market structure 0 2){0\nodepart{second}2};

\node[unknown,right=of step 4 market structure 0 2] (step 4 market structure 1 1) {1\nodepart{second}1};
\node[unknown,right=of step 4 market structure 1 1] (step 4 market structure 0 1) {0\nodepart{second}1};
\node[unknown,right=of step 4 market structure 0 1] (step 4 market structure 2 0) {2\nodepart{second}0};
\node[unknown,right=of step 4 market structure 2 0] (step 4 market structure 1 0) {1\nodepart{second}0};

\node[known bottom] at (step 4 market structure 0 2) {0\nodepart{second}2};
\node[known] at (step 4 market
structure 1 1) {1\nodepart{second}1};
\node[known bottom] at (step 4 market structure 0 1) {0\nodepart{second}1};

\node[computing top] at (step 4 market structure 2 0) {2\nodepart{second}0};


\node[step label,below=7mm of computer4] (computer5){Step 5};

\node[unknown,right=of computer5](step 5 market structure 0 2){0\nodepart{second}2};

\node[unknown,right=of step 5 market structure 0 2] (step 5 market structure 1 1) {1\nodepart{second}1};
\node[unknown,right=of step 5 market structure 1 1] (step 5 market structure 0 1) {0\nodepart{second}1};
\node[unknown,right=of step 5 market structure 0 1] (step 5 market structure 2 0) {2\nodepart{second}0};
\node[unknown,right=of step 5 market structure 2 0] (step 5 market structure 1 0) {1\nodepart{second}0};

\node[known bottom] at (step 5 market structure 0 2) {0\nodepart{second}2};

\node[known] at (step 5 market structure 1 1) {1\nodepart{second}1};

\node[known bottom] at (step 5 market structure 0 1) {0\nodepart{second}1};

\node[known top] at (step 5 market structure 2 0) {2\nodepart{second}0};

\node[computing top] at (step 5 market structure 1 0) {1\nodepart{second}0};

\node[step label,below=7mm of computer5] (computer6){Step 6};

\coordinate (final state location) at ($(step 4 market structure 0 2)!(computer3)!(step 5 market structure 0 2)$);

\node[unknown] at (final state location) (final market structure 0 2) {0\nodepart{second}2};
\node[unknown,right=of final market structure 0 2] (final market structure 1 1) {1\nodepart{second}1};
\node[unknown,right=of final market structure 1 1] (final market structure 0 1) {0\nodepart{second}1};
\node[unknown,right=of final market structure 0 1] (final market structure 2 0) {2\nodepart{second}0};
\node[unknown,right=of final market structure 2 0] (final market structure 1 0) {1\nodepart{second}0};

\node[known top] at (final market structure 2 0) {2\nodepart{second}0};
\node[known] at (final market structure 1 1) {1\nodepart{second}1};
\node[known top] at (final market structure 1 0) {1\nodepart{second}0};
\node[known bottom] at (final market structure 0 2) {0\nodepart{second}2};
\node[known bottom] at (final market structure 0 1) {0\nodepart{second}1};

\draw[-stealth,very thick,solid] (start here.south) -- (computer1.north);
\draw[-stealth,very thick,solid] (computer1.south) -- (computer2.north);
\draw[-stealth,very thick,solid] (computer2.south) -- (computer3.north);

\coordinate (a34) at ($(computer3.south)+(0,-5mm)$);
\coordinate (b34) at ($(a34)+(5cm,0)$);
\coordinate (mid34) at ($(computer6.west)+(-5mm,0)$);
\coordinate (c34) at ($(a34)!(mid34)!(b34)$);
\coordinate (d34) at ($(c34)!(computer4.west)!(mid34)$);

\draw[-stealth,very thick,solid,rounded corners] (computer3.south) -- (a34) -- (c34) -- (d34) -- (computer4.west);
\draw[-stealth,very thick,solid] (computer4.south) -- (computer5.north);
\draw[-stealth,very thick,solid] (computer5.south) -- (computer6.north);

\node[rectangle,black!90] at ($(step 1 market structure 0
2.south)+(4mm,-0.5cm)$) (calculating) {\tiny  Values being calculated}; \draw[-stealth,solid,black!90] (calculating)--
(step 1 market structure 0 2);

\node[rectangle,black!90] at ($(step 2 market structure 0 2.south)+(0,-0.5cm)$) (in hand) {\tiny Values
in hand}; \draw[-stealth,solid,black!90] (in hand) -- (step 2 market structure 0 2);

\node[rectangle,black!90] at ($(step 1 market structure 1
0.north)+(0,1cm)$) (low types) {\tiny Number of type $1$
firms}; \draw[-stealth,solid,black!90] (low types) -- (step 1
market structure 1 0);

\node[rectangle,black!90] at ($(step 1 market structure 1
0.south)+(0,-0.5)$) (high types) {\tiny Number of type $2$ firms}; \draw[-stealth,solid,black!90] (high types) --
(step 1 market structure 1 0);

\node[rectangle,black!90,draw,rounded corners,fit=(step 4 market structure 1 1)] (market structure container) {};
\node[rectangle,black!90] at ($(step 4 market structure 1 1.north)+(0,1cm)$) (market structure legend) {\tiny Market structure with one firm of each type};
\draw[-stealth,solid,black!90] (market structure legend.south) -- (market structure container);

\node[rectangle,draw,rounded corners,text centered,text width=2.8cm,font=\tiny] at ($(final market structure 0 1.south)+(0,-1.5cm)$) (step 6 instructions) {Calculate remaining entry and survival rules.};

\node at (computer6 |- step 6 instructions) (stopping point) {\pgfuseimage{atari}};
\draw[-stealth,very thick,solid] (computer6.south) -- (stopping point.north);

\end{tikzpicture}

\end{center}
{\footnotesize Note: There are five possible market structures with one or two firms. Each divided rectangle represents one of
them, with the top and bottom numbers giving the number of type 1 and type 2 firms.  Each collection of five rectangles displays the value functions being calculated (in red) and the value functions already in hand (in blue) at one stage of the algorithm. The labels mark Steps 1 through 6 of the text's construction, and the arrows indicate their sequence. The computer visually indicates that this presents an algorithm. It is an Atari 400, first sold in November 1979. It had 8 KB of RAM and could run BASIC programs loaded from a cassette tape drive.}
\end{figure}

\pgfkeyssetvalue{/homogeneous/duopoly continuation payoff}{v_S(2,2,c)-\kappa(y)}
\pgfkeyssetvalue{/homogeneous/monopoly continuation payoff}{v_S(2,0,c)-\kappa(y)}

\begin{figure}[tp]
\caption{Reduced-form Representation of a Duopoly Continuation Game\label{fig:duopolyReducedForm}}
\begin{center}
\begin{tabular}{r|c|c|}
\multicolumn{1}{c}{} & \multicolumn{1}{c}{Survive}
  & \multicolumn{1}{c}{Exit} \\ \cline{2-3}
Survive & $\stackedpayoffs{\pgfkeysvalueof{/homogeneous/duopoly continuation payoff}}{\pgfkeysvalueof{/homogeneous/duopoly continuation payoff}}$ & $\stackedpayoffs{\pgfkeysvalueof{/homogeneous/monopoly continuation payoff}}{0}$ \\ \cline{2-3}
Exit & $\stackedpayoffs{0}{\pgfkeysvalueof{/homogeneous/monopoly continuation payoff}}$ & $\stackedpayoffs{0}{0}$ \\ \cline{2-3}
\end{tabular}
\end{center}
\footnotesize\noindent Note: In each cell, the upper-left expression gives the row player's payoff. Please see the
text for details.
\end{figure}

\paragraph{Step 1: Calculation of $w_S(2,2,c)$ and $w_E(2,2,y)$.}

As depicted by the upper-left panel in Figure \ref{fig:twoTypesFakeBackwardInduction}, we start the construction by considering a market populated by two type $2$ firms. Suppose that we are at a post-entry node of the game tree with the exogenous state equal to $y=(c,z)$, and consider equilibrium survival decisions of the two incumbents. Their common survival probability $a_S(2,2,y)$ must form a symmetric Nash equilibrium of the static simultaneous-move game in Figure \ref{fig:duopolyReducedForm}. The upper-left expression in each cell is the row player's payoff. A firm that survives while its rival exits earns the monopoly post-survival value net of the cost of continuation, $v_S(2,0,c)-\kappa(y)$. Following joint survival, both firms receive the duopoly post-survival payoff net of the same cost $v_S(2,2,c)-\kappa(y)$. The following lemma places structure on this one-shot game.

\begin{lemma}
\label{lem:twoTypesHighMonopolyBeatsDuopoly}
In any equilibrium, $v_S(2,2,c)\leq v_S(2,0,c)$ for all $c\in\mathcal{C}$.
\end{lemma}
\noindent The Appendix contains the proofs of this and all other lemmas.

Lemma \ref{lem:twoTypesHighMonopolyBeatsDuopoly} allows us to prove that Figure \ref{fig:duopolyReducedForm}'s game has a unique symmetric Nash equilibrium payoff.\footnote{\label{fn:duopolyEssential}This paper's restriction to at most two active firms is essential for proving Lemma \ref{lem:twoTypesHighMonopolyBeatsDuopoly}. Within the analogous game that allows three active firms, it is not difficult to construct an equilibrium in which adding a type $2$ firm increases a type $2$ incumbent's expected discounted profits by deterring the entry of two type $1$ firms. We are considering approaches to overcoming this difficulty in our current research.} If $v_S(2,2,c)>\kappa(y)$, then survival is a strictly dominant strategy for both firms, and the unique symmetric Nash equilibrium is ``Survive, Survive.''  If instead $v_S(2,0,c)<\kappa(y)$, then exit is a strictly dominant strategy for both firms, and the unique symmetric Nash equilibrium is ``Exit, Exit.'' In the case with $v_S(2,2,c)\leq\kappa(y)\leq v_S(2,0,c)$, any symmetric Nash equilibrium's mixed strategy makes both firms indifferent between exit and continuation.\footnote{If $v_S(2,0,c)=v_S(2,2,c)=\kappa(y)$, then any survival probability forms a Nash equilibrium in mixed strategies. Assumption \ref{ass:continuousShocks} ensures that this happens with probability zero.} Bringing these three cases together establishes that $v_E(2,2,y)= \max\{0,v_S(2,2,c)-\kappa(y)\}$. Combining this result with \eqref{eq:vS} yields
\begin{equation}
\label{eq:twoTypesBellmanHighHigh}
v_S(2,2,c) =\beta\mathbb{E}\left[\pi(2,2,Y^\prime) + \max\{0,v_S(2,2,C^\prime)-\kappa(Y^\prime)\}\;\vline\;C=c\right].
\end{equation}

\noindent The right-hand side of \eqref{eq:twoTypesBellmanHighHigh} defines an operator on functions of $c$. This operator satisfies Blackwell's sufficient conditions for a contraction, and its unique fixed point determines $v_S(2,2,c)$ for every $c\in\mathcal{C}$. Although duopoly with only type $2$ firms is \emph{not} an absorbing state for the industry, we can calculate the equilibrium payoffs for that market structure without knowledge of the firms' payoffs in all possible future market structures.\footnote{Without Assumption \ref{ass:noRegress}, the right-hand side of \eqref{eq:twoTypesBellmanHighHigh} would include the post-survival continuation values from market structures with a type $1$ firm. That would make the computation of $v_S(2,2,c)$ dependent on those continuation values and so make our algorithmic approach inapplicable.} Furthermore, the calculated $v_S(2,2,c)$ gives the only possible equilibrium continuation value, because \eqref{eq:twoTypesBellmanHighHigh} provides a \emph{necessary} condition. These insights, adapted from \citet*{vsmpidTheory18:abbringetal}, give a starting point for unraveling the model's equilibrium payoffs. We set our candidate payoff $w_S(2,2,c)$ equal to this fixed point evaluated at $c$ and set
\[w_E(2,2,y)=\max\{0,w_S(2,2,c)-\kappa(y)\}.\]

\paragraph{Step 2: Calculation of $w_S(1,2,c)$, $w_E(1,2,y)$, $\alpha_E(2,y)$, and $\alpha_S(1,2,y)$.}

Moving down one panel in Figure \ref{fig:twoTypesFakeBackwardInduction}, we calculate payoffs for a type $1$ firm facing a type $2$ rival. This firm chooses survival whenever $a_S(2,1,y) v_S(1,2,c)+(1-a_S(2,1,y)) v_S(1,0,c)>\kappa(y)$. The next lemma gives a simpler expression for this condition in a natural Markov-perfect equilibrium.

\begin{lemma}
\label{lem:twoTypesLowSurvivalRule}
In a natural Markov-perfect equilibrium,
\[a_S(2,1,y) v_S(1,2,c)+(1-a_S(2,1,y)) v_S(1,0,c)>\kappa(y)\]
if and only if $v_S(1,2,c)>\kappa(y)$.
\end{lemma}

In a natural Markov-perfect equilibrium, a type $2$ rival survives whenever the type $1$ firm does. This, Lemma \ref{lem:twoTypesLowSurvivalRule}, and \eqref{eq:vE} together give $v_E(1,2,y)=\max\{0,v_S(1,2,c)-\kappa(y)\}$. Therefore, any equilibrium post-survival value $v_S(1,2,c)$ must satisfy
\begin{equation}
\label{eq:twoTypesBellmanLowHigh}
\begin{split}
&v_S(1,2,c)%
=\beta\Bigl(\Pi_{11}\mathbb{E}\left[\pi(1,2,Y^\prime)+\max\{0,v_S(1,2,C^\prime)-\kappa(Y^\prime)\}\;\vline\;C=c\right]\\
&~~~~~~~~~~~~~~~~~~~~~~~~~~~~~~~~~~~~~~~~~~~~~~~~+ \Pi_{12}\mathbb{E}\left[\pi(2,2,Y^\prime)+v_E(2,2,Y^\prime)\;\vline\;C=c\right]
\Bigr).
\end{split}
\end{equation}

\noindent After replacing $v_E(2,2,y)$ in \eqref{eq:twoTypesBellmanLowHigh} with its only possible candidate value -- calculated in Step 1 -- the right-hand side of \eqref{eq:twoTypesBellmanLowHigh} defines a contraction mapping.  We set $w_S(1,2,c)$ to equal this mapping's unique fixed point evaluated at $c$ and
\[w_E(1,2,y)=\max\{0,w_S(1,2,c)-\kappa(y)\}.\]

We set the corresponding candidate entry rule to
\[\alpha_E(2,y)  =  \true\{w_E(1,2,y)>\varphi(2,y)\},\]
and the accompanying survival rule to
\[\alpha_S(1,2,y)  = \true\{w_S(1,2,c)>\kappa(y)\},\]
where we use $\true\{\cdot\}$ as an indicator for its argument's truth. 

To see that these candidate rules equal their equilibrium counterparts with probability one, consider any equilibrium. A potential entrant facing a type $2$ incumbent enters if $w_E(1,2,y)>\varphi(2,y)$ and stays out if $w_E(1,2,y)<\varphi(2,y)$. The candidate rule $\alpha_E(2,y)$ prescribes the same actions, so the two rules can differ only where the entrant is indifferent. Indifference means $w_E(1,2,y)=\varphi(2,y)$. The entry cost is strictly positive, so an indifferent entrant has $w_E(1,2,y)>0$. Its post-entry value then equals $w_S(1,2,c)-\kappa(y)$, and indifference therefore requires $w_S(1,2,c)-\kappa(y)=\varphi(2,y)$. Similarly, $a_S(1,2,y)$ can differ from $\alpha_S(1,2,y)$ only where $w_S(1,2,c)=\kappa(y)$. For each $c$, the events $\kappa(c,Z)+\varphi(2,(c,Z))=w_S(1,2,c)$ and $\kappa(c,Z)=w_S(1,2,c)$ have probability zero by Assumption \ref{ass:continuousShocks}.

\paragraph{Step 3: Calculation of $w_S(2,0,c)$, $w_E(2,0,y)$, $w_S(2,1,c)$,  and $w_E(2,1,y)$.}

The next step, shown in the bottom-left panel of Figure \ref{fig:twoTypesFakeBackwardInduction}, calculates post-survival values for a type $2$ firm facing either no rival or a type $1$ rival. Type $1$ rivals can come and go while a type $2$ firm remains active, so these two values depend upon each other.

A type $2$ monopolist facing the exogenous state $y$ survives when $v_S(2,0,c)>\kappa(y)$ and exits when $v_S(2,0,c)<\kappa(y)$. At equality, either action is optimal, but the resulting post-entry value is zero in either case. Thus, $v_E(2,0,y)=\max\left\{0,v_S(2,0,c)-\kappa(y)\right\}$. The post-entry payoff of a type $2$ firm facing a type $1$ rival depends on $a_S(1,2,y)$.  Equations \eqref{eq:vE} and \eqref{eq:aS} therefore give
\[v_E(2,1,y)=\max\left\{0,a_S(1,2,y)v_S(2,1,c)+\left(1-a_S(1,2,y)\right)v_S(2,0,c)-\kappa(y)\right\}.\]
Next period, a surviving type $2$ monopolist faces a type $1$ entrant if the potential entrant enters and no rival otherwise. A surviving type $2$ duopolist's type $1$ rival either keeps its type or advances to type $2$. Consequently, $v_S(2,0,c)$ and $v_S(2,1,c)$ together satisfy
\begin{equation}
\label{eq:twoTypesBellmanHighMonopoly}
\begin{split}
v_S(2,0,c) =\beta\mathbb{E}\Bigl[\pi(2,0,Y^\prime)&+ a_E(2,Y^\prime)\max\Bigl\{0,a_S(1,2,Y^\prime)v_S(2,1,C^\prime)\\
&\hspace{6em}+\left(1-a_S(1,2,Y^\prime)\right)v_S(2,0,C^\prime)-\kappa(Y^\prime)\Bigr\} \\
&+\left(1-a_E(2,Y^\prime)\right)\max\Bigl\{0,v_S(2,0,C^\prime)-\kappa(Y^\prime)\Bigr\}\;\vline\;C=c\Bigr]
\end{split}
\end{equation}
and
\begin{equation}
\label{eq:twoTypesBellmanHighLow}
\begin{split}
v_S(2,1,c) = \beta \Bigl(&\Pi_{11}\mathbb{E}\Bigl[\pi(2,1,Y^\prime)+\max\Bigl\{0,a_S(1,2,Y^\prime)v_S(2,1,C^\prime)\\
&\hspace{6em}+\left(1-a_S(1,2,Y^\prime)\right)v_S(2,0,C^\prime)-\kappa(Y^\prime)\Bigr\}\;\vline\;C=c\Bigr]\\
& +\Pi_{12}\mathbb{E}[\pi(2,2,Y^\prime)+v_E(2,2,Y^\prime)\;\vline\;C=c]\Bigr).
\end{split}
\end{equation}

\noindent Step 1 established that $v_E(2,2,y)=w_E(2,2,y)$. Step 2 established that $a_E(2,(c,Z))=\alpha_E(2,(c,Z))$ and $a_S(1,2,(c,Z))=\alpha_S(1,2,(c,Z))$ with probability one. Replacing $v_E(2,2,Y^\prime)$, $a_E(2,Y^\prime)$, and $a_S(1,2,Y^\prime)$ in the right-hand sides of \eqref{eq:twoTypesBellmanHighMonopoly} and \eqref{eq:twoTypesBellmanHighLow} with their candidate counterparts therefore leaves those right-hand sides unchanged. After these replacements, the right-hand sides define a contraction mapping on pairs of functions of $c$. Its fixed point gives the only possible values for $v_S(2,0,c)$ and $v_S(2,1,c)$. Therefore, we set our candidate payoffs $w_S(2,x,c)$ (for $x\in\{0,1\}$) to this fixed point, and set
\begin{align*}
w_E(2,0,y) & =\max\{0,w_S(2,0,c)-\kappa(y)\}, \textrm{ and } \\
w_E(2,1,y) & =\max\{0,\alpha_S(1,2,y)w_S(2,1,c)+(1-\alpha_S(1,2,y))w_S(2,0,c)-\kappa(y)\}.
\end{align*}

The next lemma orders the values calculated so far. The proofs of Lemmas \ref{lem:twoTypesLowMonopolyBeatsDuopoly}, \ref{lem:twoTypesCandidateNatural}, and \ref{lem:twoTypesCalculatedMonopolyValueOrder} use this ordering.

\begin{lemma}
The calculated values $w_S(2,0,c)$, $w_S(2,1,c)$, $w_S(2,2,c)$, and $w_S(1,2,c)$ satisfy \[w_S(2,0,c)\geq w_S(2,1,c)\geq w_S(2,2,c)\geq w_S(1,2,c)\] for all $c\in\mathcal{C}$.\label{lem:twoTypesCalculatedHighValueOrder}
\end{lemma}

\paragraph{Step 4: Calculation of $w_S(1,1,c)$, $w_E(1,1,y)$, and $\alpha_E(1,y)$.}

Proceed to the upper-right panel of Figure \ref{fig:twoTypesFakeBackwardInduction} and consider a duopoly market with two type $1$ firms. The survival game played out between these firms resembles the one depicted in Figure \ref{fig:duopolyReducedForm}.  The next lemma extends Lemma \ref{lem:twoTypesHighMonopolyBeatsDuopoly}'s ordering of monopoly and duopoly continuation values to type $1$ firms, which again yields a unique symmetric Nash equilibrium payoff.

\begin{lemma}
\label{lem:twoTypesLowMonopolyBeatsDuopoly}
In any equilibrium, \[v_S(1,0,c)\geq v_S(1,1,c)\] for all $c\in\mathcal{C}$.
\end{lemma}

With $v_S(1,0,c)\geq v_S(1,1,c)$, the survival game between two type $1$ firms has the same three cases as the game between two type $2$ firms in Step 1. If $v_S(1,1,c)>\kappa(y)$, survival is strictly dominant. If $v_S(1,0,c)<\kappa(y)$, exit is strictly dominant. Otherwise, any symmetric equilibrium leaves both firms indifferent, with payoff zero. Hence $v_E (1,1,y)  = \max\left\{0,v_S(1,1,c)-\kappa(y)\right\}$ in any equilibrium, and equilibrium continuation values satisfy
\begin{equation} \label{eq:twoTypesBellmanLowLow}
\begin{split}
v_S (&1,1,c) =  \beta\Pi^2_{11}\mathbb{E}\left[\pi(1,1,Y^\prime)+\max\left\{0,v_S(1,1,C^\prime)-\kappa(Y^\prime)\right\}\;\vline\;C=c\right] \\
	& +\beta\Pi_{11}\Pi_{12}\mathbb{E}\left[\pi(1,2,Y^\prime)+v_E(1,2,Y^\prime)+\pi(2,1,Y^\prime)+v_E(2,1,Y^\prime)\;\vline\;C=c\right] \\
	 & +\beta\Pi^2_{12}\mathbb{E}\left[\pi(2,2,Y^\prime)+v_E(2,2,Y^\prime)\;\vline\;C=c\right].
\end{split}
\end{equation}

\noindent The right-hand side of \eqref{eq:twoTypesBellmanLowLow} contains three post-entry values: $v_E(1,2,Y^\prime)$, $v_E(2,2,Y^\prime)$, and $v_E(2,1,Y^\prime)$. Steps 1 and 2 showed that the first two equal their candidate counterparts, $w_E(1,2,y)$ and $w_E(2,2,y)$, for all $y\in\mathcal{Y}$. The third is the value of a type $2$ firm facing a type $1$ rival. Step 3's expressions for it and for its candidate counterpart $w_E(2,1,y)$ have the same form. The first weights $v_S(2,1,c)$ and $v_S(2,0,c)$ with the rival's equilibrium survival rule $a_S(1,2,y)$. The second weights $w_S(2,1,c)$ and $w_S(2,0,c)$ with the candidate rule $\alpha_S(1,2,y)$. Step 3 set these post-survival values equal, so the two expressions can differ only where $a_S(1,2,y)$ differs from $\alpha_S(1,2,y)$. Step 2 showed that these rules agree with probability one. Hence $v_E(2,1,(c,Z))=w_E(2,1,(c,Z))$ with probability one. Replacing the three equilibrium values in the right-hand side of \eqref{eq:twoTypesBellmanLowLow} with their candidate counterparts therefore leaves it unchanged. After these replacements, the only undetermined function on the right-hand side is $c\mapsto v_S(1,1,c)$. The resulting operator is a contraction with this function as its unique fixed point. We set $w_S(1,1,c)$ to this fixed point evaluated at $c$ and set
\[ w_E(1,1,y) = \max\{0,w_S(1,1,c)-\kappa(y)\}\]
and
\[\alpha_E(1,y)=\true\{w_E(1,1,y)>\varphi(1,y)\}. \]
We now show that this candidate rule equals its equilibrium counterpart with probability one. Because $v_S(1,1,c)=w_S(1,1,c)$, entering a type $1$ monopoly is worth $w_E(1,1,y)$ in any equilibrium. The rules $a_E(1,y)$ and $\alpha_E(1,y)$ can therefore differ only where the entrant is indifferent. As in Step 2, the strictly positive entry cost makes indifference require $w_S(1,1,c)-\kappa(y)=\varphi(1,y)$. For each $c$, the event $\kappa(c,Z)+\varphi(1,(c,Z))=w_S(1,1,c)$ has probability zero by Assumption \ref{ass:continuousShocks}.

\paragraph{Step 5: Calculation of $w_S(1,0,c)$ and $w_E(1,0,y)$.}

As shown in the middle-right panel of Figure \ref{fig:twoTypesFakeBackwardInduction}, the only remaining unconstructed payoff is that for a type $1$ monopolist. Such a firm survives when $v_S(1,0,c)>\kappa(y)$ and exits when $v_S(1,0,c)<\kappa(y)$, so $v_E(1,0,y) =  \max\left\{0,v_S(1,0,c)-\kappa(y)\right\}$. If it survives, its own type transition and the potential entrant's decision together place it into one of four market structures in the next period. Consequently,
\begin{align}
 v_S(1,0,c) = &\beta\Pi_{11}\mathbb{E}\Bigl[\pi(1,0,Y^\prime)+a_E(1,Y^\prime)v_E(1,1,Y^\prime) \label{eq:twoTypesBellmanLowMonopoly}\\
& +\left(1-a_E(1,Y^\prime)\right)\max\left\{0,v_S(1,0,C^\prime)-\kappa(Y^\prime)\right\}\;\vline\;C=c\Bigr]\nonumber\\
&+\beta\Pi_{12}\mathbb{E}\Bigl[\pi(2,0,Y^\prime)+a_E(2,Y^\prime)v_E(2,1,Y^\prime) \nonumber\\
 & +\left(1-a_E(2,Y^\prime)\right)v_E(2,0,Y^\prime)\;\vline\;C=c\Bigr].\nonumber
\end{align}

\noindent Steps 1--4 established that $v_E(1,1,y)=w_E(1,1,y)$ and $v_E(2,0,y)=w_E(2,0,y)$ for all $y\in\mathcal{Y}$ and that $v_E(2,1,(c,Z))=w_E(2,1,(c,Z))$ with probability one. Steps 2 and 4 also established that $a_E(1,(c,Z))=\alpha_E(1,(c,Z))$ and $a_E(2,(c,Z))=\alpha_E(2,(c,Z))$ with probability one. Replacing these equilibrium values and rules in the right-hand side of \eqref{eq:twoTypesBellmanLowMonopoly} with their candidate counterparts therefore leaves it unchanged. After these replacements, the right-hand side defines a contraction mapping with $c\mapsto v_S(1,0,c)$ as its unique fixed point. We set $w_S(1,0,c)$ equal to this fixed point evaluated at $c$ and set
\[ w_E(1,0,y) = \max\{0,w_S(1,0,c)-\kappa(y)\}.\]

\paragraph{Step 6: Calculation of $\alpha_S(2,0,y)$, $\alpha_S(1,0,y)$, $\alpha_S(2,1,y)$, $\alpha_E(0,y)$, $\alpha_S(2,2,y)$ and $\alpha_S(1,1,y)$.}

The previous five steps determined all candidate payoffs. All that remains is to determine the candidate's survival rules for firms facing rivals with weakly lower types and its rule for entry into an empty market. We begin with three rules that are best responses to the candidate payoffs and to candidate rules already in hand.

The candidate's survival rule for a monopolist of type $k\in\{1,2\}$ is
\[\alpha_S(k,0,y)=\true\{w_S(k,0,c)>\kappa(y)\}.\]
Its survival rule for a type $2$ firm facing a type $1$ rival takes the rival's rule $\alpha_S(1,2,y)$ from Step 2 as given:
\begin{equation}
\label{eq:twoTypesSurvivalRuleHighLow}
\alpha_S(2,1,y)=\true\left\{\alpha_S(1,2,y)w_S(2,1,c)
+\left(1-\alpha_S(1,2,y)\right)w_S(2,0,c)>\kappa(y)\right\}.
\end{equation}
Its rule for entry into an empty market takes as given the rule $\alpha_E(1,y)$ from Step 4, which the second potential entrant follows:
\[\alpha_E(0,y)=\true\left\{\alpha_E(1,y)w_E(1,1,y)
+\left(1-\alpha_E(1,y)\right)w_E(1,0,y)>\varphi(0,y)\right\}.\]

The remaining rule governs two firms of the same type $k\in\{1,2\}$. Each responds to the other's use of the same rule, so we set it to a symmetric Nash equilibrium of the static game in Figure \ref{fig:duopolyReducedForm}, or of its analog for two type $1$ firms, with the candidate values as payoffs:
\[
\alpha_S(k,k,y) = \left\{\begin{array}{ll}
1&\text{if }w_S(k,k,c)>\kappa(y),\\
\frac{w_S(k,0,c)-\kappa(y)}{w_S(k,0,c)
        -w_S(k,k,c)}& \text{if }w_S(k,k,c)\leq \kappa(y) \text{ and }w_S(k,0,c)>\kappa(y),\text{ and}\\
0&\text{otherwise.}
\end{array}\right.
\]
This rule has both firms survive when duopoly survival is profitable, mix so that each is indifferent when only monopoly survival is profitable, and exit otherwise. This completes the construction of the candidate strategy.

Two tasks remain. The first is to show that the candidate strategy forms a natural Markov-perfect equilibrium. The second is to show that it is the only one, up to events of probability zero. Steps 1 through 5 began the second task: they showed that any equilibrium's post-survival values equal the candidate's.

\paragraph{Existence.}

The candidate strategy satisfies Definition \ref{def:naturalEquilibrium}'s naturality restriction.
\begin{lemma}\label{lem:twoTypesCandidateNatural}
The calculated $\alpha_S(1,2,y)$ and $\alpha_S(2,1,y)$ satisfy
\[ \alpha_S(1,2,y)>0 \implies \alpha_S(2,1,y)=1\] for all $y\in\mathcal{Y}$.
\end{lemma}

It remains to show that no firm gains from a one-shot deviation from the candidate strategy, so that it satisfies \eqref{eq:aE} and \eqref{eq:aS}. Every candidate rule but one is a best response by construction.

The exception is the survival rule of a type $1$ firm facing a type $2$ rival. Step 2 set $\alpha_S(1,2,y)=\true\{w_S(1,2,c)>\kappa(y)\}$. This rule compares the continuation cost with $w_S(1,2,c)$, the firm's value when its rival survives. The firm's best response compares that cost with a weighted average of $w_S(1,2,c)$ and $w_S(1,0,c)$. The weights are the probabilities that the rival's candidate rule assigns to survival and exit:
\[ \tilde{\alpha}_S(1,2,y)=\true\{\alpha_S(2,1,y)w_S(1,2,c)+\left(1-\alpha_S(2,1,y)\right)w_S(1,0,c)>\kappa(y)\}.\]
We now show that $\tilde{\alpha}_S(1,2,y)=\alpha_S(1,2,y)$. This requires one more ordering of the candidate values: a type $1$ monopolist is worth at most as much as a type $2$ monopolist.
\begin{lemma}
\label{lem:twoTypesCalculatedMonopolyValueOrder}
The calculated $w_S(1,0,c)$ and $w_S(2,0,c)$ satisfy
\[ w_S(1,0,c)\leq w_S(2,0,c)\]
for all $c\in\mathcal{C}$.
\end{lemma}

\noindent This order rests on Assumption \ref{ass:entryBarriers}. If entry were cheaper against a type $2$ incumbent than against a type $1$ incumbent, a type $2$ monopolist could attract more entry, and a type $1$ monopolist might then be the more valuable of the two.

There are two cases. If $\alpha_S(2,1,y)=1$, the rival survives for sure, and $\tilde{\alpha}_S(1,2,y)=\true\{w_S(1,2,c)>\kappa(y)\}=\alpha_S(1,2,y)$. If instead $\alpha_S(2,1,y)=0$, then Lemma \ref{lem:twoTypesCandidateNatural} gives $\alpha_S(1,2,y)=0$. Equation \eqref{eq:twoTypesSurvivalRuleHighLow} in turn gives $w_S(2,0,c)\leq\kappa(y)$. Lemma \ref{lem:twoTypesCalculatedMonopolyValueOrder} extends this to $w_S(1,0,c)\leq\kappa(y)$. A type $1$ firm that survived would be a monopolist unable to cover its continuation cost, so $\tilde{\alpha}_S(1,2,y)=0=\alpha_S(1,2,y)$.

The candidate strategy is therefore a best response to itself, and it forms a natural Markov-perfect equilibrium. The candidate payoffs $w_E$ and $w_S$ are its continuation values.

\paragraph{Essential uniqueness.}

Consider any equilibrium. Steps 1 through 5 showed that its payoffs equal the candidate's. Its rules and the candidate's are therefore best responses to the same payoffs, and they can differ only where a firm is indifferent between its actions. Because the shock is continuously distributed and entry costs are strictly positive, a firm facing a rival of a different type, or no rival, is indifferent with probability zero. Two firms of the same type can both be indifferent with positive probability. Steps 1 and 4 showed that their survival game then has exactly one symmetric equilibrium, in which both use the mixing probability $\alpha_S(k,k,y)$.

Any equilibrium strategy therefore equals the candidate strategy except on sets of states with probability zero. This is the sense in which the two-type game's equilibrium is essentially unique.


\section{Equilibrium Construction and Uniqueness\label{sec:constructionUniqueness}}
With the intuition from the case with $\check{k}=2$ in hand, we proceed to allow for an arbitrary $\check{k}\geq 1$. Our analysis centers on our algorithm for equilibrium calculation. 

\begin{alg}[Equilibrium Calculation]
\label{alg:candidateEquilibrium} Compute a candidate equilibrium strategy $(\alpha_E,\alpha_S)$ and payoffs $w_E$ and $w_S$ in
two stages.
\begin{enumerate}
\item Use Procedure \ref{proc:fakeBackwardInduction} to compute $w_E$, $w_S$, $\alpha_E(h,y)$ for all $(h,y)\in\mathbb{K}\times\mathcal{Y}$, and $\alpha_S(l,h,y)$ for all $(l,h,y)\in\mathbb{K}^2\times\mathcal{Y}$ such that $ l<h$.
\item Use the following procedure to compute the remaining components of the candidate strategy:
\begin{itemize}
\item for all $y\in\mathcal{Y}$, set 
\[\alpha_E(0,y)\leftarrow \true\left\{\alpha_E(1,y)w_E(1,1,y)
+\left(1-\alpha_E(1,y)\right)w_E(1,0,y)>\varphi(0,y)\right\};\]
\item for all $h\in\mathbb{K}$ and $y\in\mathcal{Y}$, set
\[\alpha_S(h,0,y)\leftarrow \true\left\{w_S(h,0,c)>\kappa(y)\right\};\]
\item for all $h\in\mathbb{K}$ and $y\in\mathcal{Y}$, set
\[
\alpha_S(h,h,y)\leftarrow
\begin{cases}
1 & \text{if }w_S(h,h,c)>\kappa(y),\\
\dfrac{w_S(h,0,c)-\kappa(y)}{w_S(h,0,c)-w_S(h,h,c)}
& \text{if }w_S(h,h,c)\leq\kappa(y)<w_S(h,0,c),\\
0 & \text{otherwise};
\end{cases}
\]
and
\item for all $(h,l,y)\in\mathbb{K}^2\times\mathcal{Y}$ with $l<h$,
\[\alpha_S(h,l,y)\leftarrow \true\left\{\alpha_S(l,h,y)w_S(h,l,c)
+\left(1-\alpha_S(l,h,y)\right)w_S(h,0,c)>\kappa(y)\right\}.\]
\end{itemize}
\end{enumerate}
\end{alg}

Procedure \ref{proc:fakeBackwardInduction} presents the algorithm's first stage as a flow chart. This stage applies FBI to compute all payoffs, the survival rule for a duopolist facing a rival with a strictly higher technology type, and the entry rules for potential entrants facing an incumbent. It begins by setting the candidate strategy $(\alpha_E,\alpha_S)$ and the payoffs $w_E$ and $w_S$ to arbitrary values. Two nested loops follow. The outer loop indexes the weakly better firm's technology type by $h$, which counts down from $\check{k}$ to $1$. Within each turn of the outer loop, an inner loop indexes the weakly worse firm's type by $l$, which counts down from $h$ to $1$. Each turn of the inner loop computes $w_S(l,h,c)$ as the fixed point of an operator $T_{lh}$. When $l<h$, it also computes the worse firm's survival rule. After the inner loop ends, the outer loop computes the entry rule $\alpha_E(h,y)$ for a potential entrant facing a type $h$ incumbent. It then computes the values $w_S(h,x,c)$ of a type $h$ firm facing no rival or a worse rival as a single fixed point of an operator $T_h$.


\begin{procedure}
\begin{center}

\vspace{-1cm}

\begin{tikzpicture}[start chain=going below, node distance=0.4cm and 1cm,
     start/.style={circle, white , fill=green!50!black, double=green!50!black, draw, very thick, text width=4.5em, text centered, scale=0.8},
    stop/.style={chamfered rectangle, white, fill=red, double=red, draw, very thick,inner ysep=1em, inner xsep=0.5em, scale=0.8},
    decision/.style={diamond,draw=blue, thick, fill=blue!20,text width=4.5em, text badly centered, inner sep=1pt, scale=0.8},
    io/.style={tape,draw=red!60,fill=red!30, text width=20em,text centered, minimum height=2em, scale=0.8},
    block/.style = {rectangle, draw=blue, thick, fill=blue!20, text width=7em, text centered, rounded corners, minimum height=2em, scale=0.8},
    medium block/.style = {rectangle, draw=blue, thick, fill=blue!20, text width=22em, text centered, rounded corners, minimum height=2em, scale=0.8},
    large block/.style = {rectangle, draw=blue, thick, fill=blue!20, text width=12em, text centered, rounded corners, minimum height=4em, scale=0.8},
    biggie block/.style = {rectangle, draw=blue, thick, fill=blue!20, text width=24em, text centered, rounded corners, minimum height=3.5em, scale=0.8},
    super biggie block/.style = {rectangle, draw=blue, thick, fill=blue!20, text width=35em, text centered, rounded corners, minimum height=3.5em, scale=0.8},
    line/.style = {draw, thick, -latex',rounded corners,shorten >=2pt, scale=0.8},
    information text/.style={rounded corners,fill=black!10,inner sep=1ex, text width=17em, scale=0.8, text justified}],

  \node [on chain, start] (start) {\textsc{START}};
  
  \node [on chain, super biggie block] (outer loop initialization) {$\alpha_S(k,x,y) \leftarrow 1, ~\alpha_E(x,y)\leftarrow 0, ~w_S(k,x,c)\leftarrow 0, ~w_E(k,x,y)\leftarrow 0, ~h\leftarrow \check{k}$};
  \node [on chain, block] (inner loop initialization) {$l\leftarrow h$};
  \node[on chain, medium block] (compute wL){\vspace{-.5cm}\begin{eqnarray*}
   f^\star(c)& \leftarrow & 0 \\
 f^\star(c) &\leftarrow & \lim_{r\rightarrow\infty} T^r_{lh}(f^\star)(c)\\
  w_S(l,h,c) & \leftarrow & f^\star(c)\\
  w_E(l,h,y) & \leftarrow & \max\left\{0,f^\star(c)-\kappa(y)\right\}
  \end{eqnarray*}};
  \node[on chain, decision] (test survival rule calculation) {$l<h$?};
    \begin{scope}[start branch=test survival rule calculation going right]
    \node[on chain, large block] (survival rule L) {$\alpha_S(l,h,y)\leftarrow \true\left\{w_S(l,h,c)>\kappa(y)\right\}$};
    \end{scope}
   \node[on chain, decision] (test inner loop) {$l=1$?};
    \begin{scope}[start branch=test inner loop going left, node distance=0.5cm and 2.2cm]
    \node[on chain, block] (inner loop decrement) {$l\leftarrow l-1$};
    \end{scope}
   \node[on chain, biggie block] (aE calculations) {\vspace{-.5cm}{$\alpha_E(h,y) \leftarrow \true\left\{w_E(1,h,y)>\varphi(h,y)\right\}$}};
   \node[on chain, super biggie block] (H calculations) {\vspace{-.5cm}\begin{eqnarray*}
   f^\dagger(x,c)& \leftarrow & 0\quad\textrm{for }0\leq x<h\\
	f^\dagger(x,c) & \leftarrow & \lim_{r\rightarrow\infty} T^r_h(f^\dagger)(x,c) \\
   	w_S(h,x,c) & \leftarrow & f^\dagger(x,c) ~ \textrm{for} ~ 0\leq x<h \\
	w_E(h,0,y) & \leftarrow & \max\left\{0,f^\dagger(0,c)-\kappa(y)\right\} \\
	w_E(h,x,y) & \leftarrow & \max\left\{0,\alpha_S(x,h,y)f^\dagger(x,c)\right. \\
	 & & \left.+\left(1-\alpha_S(x,h,y)\right)f^\dagger(0,c)-\kappa(y)\right\}\textrm{ for }0< x < h
   	\end{eqnarray*}};
    \node[on chain, decision] (test outer loop) {$h=1$?};
    \begin{scope}[start branch=test outer loop going left, node distance=0.5cm and 2.2cm]
    \node[on chain, block] (outer loop decrement) {$h\leftarrow h-1$};
    \end{scope}
    \node [on chain, io] (io) {Return $w_S$, $w_E$, $\alpha_E$, and $\alpha_S(l,h,y)$ for all $(l,h,y)\in\mathbb{K}^2\times\mathcal{Y}$ with $l<h$.};
\begin{scope}[start branch=io going right, node distance= 3cm and 1cm]
\node[on chain, information text]{The text defines $T_{lh}$ and $T_h$. Here, ``$\leftarrow$'' indicates assignment, $T^r$ denotes $T$ composed with itself $r$ times, and $\lim_{r\rightarrow\infty} T^r$ is calculated by iterating on $T$ until convergence up to some prespecified tolerance. Assignments hold for all values of their unbound arguments.};
\end{scope}
    \node [on chain, stop] (stop) {\textsc{STOP}};
     \begin{scope}[every path/.style=line]
      \path (start) -- (outer loop initialization);
      \path (outer loop initialization) -- (inner loop initialization);
      \path (inner loop initialization) -- (compute wL);
      \path (compute wL) -- (test survival rule calculation);
      \path (test survival rule calculation.east) -- node[pos=0.4,above] {Yes} (survival rule L);
      \path (test survival rule calculation.south) --  node[pos=0.4,left] (helper3) {No} (test inner loop.north);
      \path (survival rule L.south) |- (test inner loop.east);
      \path (test inner loop.south) -- node[pos=0.4,left] {Yes} (aE calculations);
      \path (aE calculations.south) -- (H calculations);
      \path (test inner loop.west) -- node[pos=0.4,above] {No} (inner loop decrement);
      \path (inner loop decrement.north) |- (compute wL.west);
      \path (H calculations) -- (test outer loop);
      \path (test outer loop.south) -- node[pos=0.4,left] {Yes} (io);
      \path (test outer loop.west) -- node[pos=0.4,above] {No} (outer loop decrement);
      \path (io) -- (stop);

      \coordinate (helper1) at ($(outer loop decrement.west)-(0.2cm, 0cm)$);
      \path (outer loop decrement.west) -- (helper1) |- (inner loop initialization.west);
     \end{scope}

\end{tikzpicture}

\end{center}
\vspace{-5mm}
\caption[FBI Stage of the Equilibrium Calculations]{FBI Stage of the Equilibrium Calculations\label{proc:fakeBackwardInduction}}
\medskip

\end{procedure}

On the $(h,l)$ iteration of the inner loop, the post-survival value function $c\mapsto w_S(l,h,c)$ of a type $l$ firm facing a type $h$ rival is computed as the unique fixed point of the operator $T_{lh}$ defined by
\begin{align*}
T_{lh}(f)(c)   \equiv 
&\beta\mathbb{E}\Bigg[\sum_{i=l}^{\check{k}}\sum_{j=h}^{\check{k}}  \Pi_{li}\Pi_{hj}\Bigg(\pi(i,j,Y^\prime) +\true\left\{(i,j)\neq(l,h)\right\} w_E (i,j,Y^\prime) \Bigg)\\
&+\Pi_{hh}\Pi_{ll}\max\left\{0,f(C^\prime)-\kappa(Y^\prime)\right\}\;\vline\;C=c\Bigg].
\end{align*}

After the inner loop is complete, the outer loop proceeds to calculate the post-survival value function $(x,c)\mapsto w_S(h,x,c)$ on $\{0,\ldots,h-1\}\times\mathcal{C}$ for a type $h$ firm facing either no rival ($x=0$) or a rival with strictly inferior type ($0<x<h$) as the unique fixed point of the operator $T_{h}$ defined as follows.
For $h=1$, set
\[
\begin{split}
T_1(f)(0,c)\equiv\beta\mathbb{E}\Bigg[
&\sum_{i=1}^{\check{k}}\Pi_{1i}\pi(i,0,Y^\prime)\\
&+\sum_{i=2}^{\check{k}}\Pi_{1i}\Big(\alpha_E(i,Y^\prime)w_E(i,1,Y^\prime)
+\left(1-\alpha_E(i,Y^\prime)\right)w_E(i,0,Y^\prime)\Big)\\
&+\Pi_{11}\Big(\alpha_E(1,Y^\prime)w_E(1,1,Y^\prime)\\
&\qquad+\left(1-\alpha_E(1,Y^\prime)\right)
\max\left\{0,f(0,C^\prime)-\kappa(Y^\prime)\right\}\Big)
\;\vline\;C=c\Bigg].
\end{split}
\]
For $h>1$, set
\[
\begin{split}
T_h(f)(0,c)\equiv\beta\mathbb{E}\Bigg[
&\sum_{i=h}^{\check{k}}\Pi_{hi}\pi(i,0,Y^\prime)\\
&+\sum_{i=h+1}^{\check{k}}\Pi_{hi}\Big(\alpha_E(i,Y^\prime)w_E(i,1,Y^\prime)
+\left(1-\alpha_E(i,Y^\prime)\right)w_E(i,0,Y^\prime)\Big)\\
&+\Pi_{hh}\Bigg(\alpha_E(h,Y^\prime)\max\Big\{0,
\alpha_S(1,h,Y^\prime)f(1,C^\prime)\\
&\qquad+\left(1-\alpha_S(1,h,Y^\prime)\right)f(0,C^\prime)-\kappa(Y^\prime)\Big\}\\
&\qquad+\left(1-\alpha_E(h,Y^\prime)\right)
\max\left\{0,f(0,C^\prime)-\kappa(Y^\prime)\right\}\Bigg)
\;\vline\;C=c\Bigg],
\end{split}
\]
and 
\[
\begin{split}
T_{h}(f)&(x,c)\equiv \\
 \beta \mathbb{E} \Bigg[&\sum_{i=h}^{\check{k}}\sum_{j=x}^{\check{k}}\Pi_{hi}\Pi_{xj}\pi(i,j,Y^\prime) 
+\sum_{j=h}^{\check{k}}\Pi_{hh}\Pi_{xj}w_E(h,j,Y^\prime) +\sum_{i=h+1}^{\check{k}}\sum_{j=x}^{\check{k}}\Pi_{hi}\Pi_{xj}w_E(i,j,Y^\prime) \\ 
 + & \sum_{j=x}^{h-1}\Pi_{hh}\Pi_{xj} \max\left\{0,\alpha_S(j,h,Y^\prime)f(j,C^\prime) + \left(1-\alpha_S(j,h,Y^\prime)\right)f(0,C^\prime)-\kappa(Y^\prime)\right\}\;\vline\;C=c\Bigg].
\end{split}
\]
\noindent for $x\in\{1,\ldots,h-1\}$.  Both $T_{lh}$ and $T_h$ satisfy Blackwell's sufficient conditions with modulus $\beta$, so each has exactly one fixed point.

The algorithm's second stage calculates the survival rules for monopolists and for duopoly incumbents facing a rival with a \emph{weakly} lower technology type and the entry rule for a potential entrant facing an empty market. This mimics the final step of the case with $\check{k}=2$. As in that example, the survival rule is pure unless both firms have the same type, the payoff to duopoly continuation is negative, and the payoff to monopoly continuation is positive. By construction, the resulting candidate equilibrium probability of survival in this case lies in $[0,1)$. 

Verification that Algorithm \ref{alg:candidateEquilibrium}'s candidate strategy forms a natural Markov-perfect equilibrium requires us to determine that its corresponding values are those obtained when all firms follow the candidate, that no firm can benefit from a one-shot deviation from the candidate when all other firms follow it, and that the candidate never directs a firm to exit when a lower-type competitor survives with positive probability. Online Supplement S1 \citep*{duopolysupp:abbringetal} verifies that it satisfies these requirements in the proof of the following theorem.
\begin{theorem}[Algorithmic Functionality and Existence]
\label{th:existence}
Algorithm \ref{alg:candidateEquilibrium} computes an equilibrium and its associated continuation values. 
\end{theorem}

The analysis of the case with $\check{k}=2$ and strictly positive entry costs demonstrated that, in any equilibrium, a firm's post-survival value as a monopolist weakly exceeds its post-survival value facing a rival of weakly higher type: $v_S(k,0,c)\geq v_S(k,x,c)$ for all $k,x\in\{1,2\}$ with $k\leq x$ and all $c\in\mathcal{C}$. These inequalities made the fixed-point equations used to calculate $w_S(2,2,c)$, $w_S(1,2,c)$, and $w_S(1,1,c)$ necessary conditions, so their solutions equaled the only possible equilibrium post-survival values for a duopolist facing a rival of weakly higher type. For arbitrary $\check{k}$, Online Supplement S1 establishes, in particular, that the candidate post-survival value $w_S(k,x,c)$ weakly decreases with the rival's type $x$. Building down recursively from the highest type, Online Supplement S2 uses this property, together with uniqueness already established for higher types, to obtain $v_S(k,0,c)\geq v_S(k,x,c)$ for every pair of types $k\leq x$ in any equilibrium. These inequalities in turn make the relevant fixed-point equations necessary and yield

\begin{theorem}[Essential Uniqueness]
\label{th:uniqueness}
Suppose that $\varphi(x,y)>0$ for all $(x,y)\in(\{0\}\cup\mathbb{K})\times\mathcal{Y}$. Then, any equilibrium $(a_E,a_S)$ and its continuation values $(v_E,v_S)$ are almost surely equal to the candidate strategy $(\alpha_E,\alpha_S)$ and continuation values $(w_E,w_S)$ computed by Algorithm \ref{alg:candidateEquilibrium}. In particular, $v_S(k,x,c)=w_S(k,x,c)$ for all $(k,x,c)\in\mathbb{K}\times(\{0\}\cup\mathbb{K})\times\mathcal{C}$. Moreover, for all $(k,x,c)\in\mathbb{K}\times(\{0\}\cup\mathbb{K})\times\mathcal{C}$, $v_E(k,x,(c,Z))=w_E(k,x,(c,Z))$ and $a_S(k,x,(c,Z))=\alpha_S(k,x,(c,Z))$ with probability one; and, for all $(x,c)\in(\{0\}\cup\mathbb{K})\times\mathcal{C}$, $a_E(x,(c,Z))=\alpha_E(x,(c,Z))$ with probability one.
\end{theorem}

\noindent Given any joint distribution of the exogenous state $Y_{t_0}$ and market structure (the number of active firms and their technology types) at some time $t_0$ such that Assumption \ref{ass:continuousShocks} holds, Theorem \ref{th:uniqueness} implies that all equilibria induce the same law for the exogenous-state and market-structure process from $t_0$ onward.

For completeness, both Theorems \ref{th:existence} and \ref{th:uniqueness} include the homogeneous-firms case $\check{k}=1$. \citet*{vsmpidTheory18:abbringetal} show that our results extend beyond duopoly to oligopoly with an arbitrary number of firms if $\check{k}=1$. The type of heterogeneity we allow for here enables richer economic mechanisms, such as high-type firms deterring entry by low-type firms. This changes the analysis substantially. Indeed, Theorem \ref{th:uniqueness} does {\em not} directly extend to oligopoly with an arbitrary number of firms if $\check{k}>1$ (see Footnote \ref{fn:duopolyEssential}).

Like Section \ref{sec:twoTypes}'s two-type example, Theorem \ref{th:uniqueness} requires that entry costs are strictly positive. This ensures that a potential entrant strictly prefers staying out whenever its post-entry value is zero. If entry is costless instead, a potential entrant with a zero post-entry value would be indifferent between entering and staying out. Theorem \ref{th:existence}, which allows zero entry costs, still gives an equilibrium. However, there may now be alternative equilibria that differ from Theorem \ref{th:existence}'s equilibrium with positive probability.

First, there may be equilibria in which firms enter in states in which they subsequently exit for sure. As these states may have positive probability, the entry rules in such equilibria may not equal those from Theorem \ref{th:existence} with probability one. Nevertheless, because entry followed by sure exit is equivalent to no entry from both the potential entrant's and the incumbent's perspective, this does not affect continuation values and leads to the same post-survival outcomes.

Second, there may be equilibria in which firms enter in states in which the post-entry value is zero because they subsequently mix nontrivially between survival and exit (as entrants get type $1$ and equilibrium mixing only occurs in same-type duopolies, this only happens with entry into a type $1$ monopoly). Because the entrant then survives with positive probability, its entry changes the incumbent's continuation value. As states in which firms mix may have positive probability, equilibrium strategies and continuation values may not equal Theorem \ref{th:existence}'s with probability one.

There is nonetheless good reason to adopt Theorem \ref{th:existence}'s equilibrium even if entry costs are zero. We conjecture that this equilibrium is the limit of the essentially unique (by Theorem \ref{th:uniqueness}) equilibria at strictly positive entry costs as those costs shrink to zero. Intuitively, continuity of the fixed-point calculations and smoothing by the payoff shocks $Z$ suggest that the resulting equilibrium values and strategies converge to those calculated by Algorithm \ref{alg:candidateEquilibrium} at zero entry costs. Therefore, in the next section's application, we use Theorem \ref{th:existence}'s equilibrium for all specifications, including those with zero entry costs.


\section{A Quantitative Analysis of Competition Policy}
\label{sec:competitionPolicy}

In this penultimate section, we exemplify the duopoly game's application with a quantitative analysis of competition policy in a new market. Although the parameter values we use have empirical interpretations, our results characterize no particular real-world market.  We calculate Theorem \ref{th:existence}'s equilibrium outcome, which we call the outcome of \emph{laissez faire} competition policy. We compare it with the choices of a social planner who controls the entry, exit, and conduct of two identically endowed firms.  We finish by characterizing competition policies that narrow the gap between the planner's expected surplus and that realized under laissez faire. 

\subsection{The Example's Specification}

The unit of time equals one month, and we set the discount factor $\beta$ to $\parameter{\grossAnnualFactor}^{-1/12} \approx \parameter{\discountFactor}$. This matches an annual cost of capital equal to \counted{\annualCostOfCapital} percent. The market begins with one customer. At first, the market size (which we measure in customers) is \emph{immature}. In this state, the market grows by \counted{\growthPercent} percent with probability \probability{\growthProbability}. Each month, the immature phase ends with probability $1/\counted{\developmentMonths}$, so that it lasts $\counted{\developmentMonths}$ months on average. These two draws are independent of each other. The persistent state $C_t$ is therefore the pair $(S_t,D_t)$, with $S_t$ the market's size and the Boolean $D_t$ indicating that the market is immature. Market size stops growing either when the market reaches maturity or when it reaches the maximum possible market size by making \counted{\the\numexpr\marketSizes-1\relax} steps. Both components take finitely many values, so Assumption \ref{ass:finitePersistentState} holds good. Customers enjoy both the good in question -- widgets -- and residual income. Their utility is quadratic and concave in widgets, linear in residual income, and additively separable between the two goods. Therefore, each customer's inverse demand curve is $P=\bar{p}-b q$ with both $\bar{p}$ and $b$ positive. The constant $\bar{p}$ is the customer's choke price. All customers are identical, so the aggregate inverse demand curve is $P=\bar{p}-(b/S)Q$, where $Q$ is industry output and $S$ is market size.

Firms produce identical widgets, and quadratic cost functions represent their technologies. A firm of type $k\in\mathbb{K}$ has costs $c_k(q)=\mu_k q + \frac{\omega}{2} q^2$ with $\mu_1\geq \mu_2\geq \ldots \geq \mu_{\check{k}}$. The curvature $\omega$ is positive, so a firm's marginal cost rises with its own output. Its output and its profit therefore approach finite limits as the market grows, rather than rising in proportion with it. For our experiments, we normalize the choke price $\bar{p}$ to equal \counted{\demandIntercept} and set $\mu_1=\bar{p}$, so that a new firm has marginal cost of its first unit equal to the choke price. For each remaining type $k>1$, $\log\mu_k=\log\mu_{k-1}-\chosen{\costDecline}$. A firm advances one type with probability $\probability{\advanceProbability}$ per month. Otherwise it stays in place, so our example satisfies Assumptions \ref{ass:monotoneTechnology} and \ref{ass:noRegress}. We set $\check{k}=\counted{\technologyTypes}$, so a firm operating the best technology has marginal cost of its first unit equal to \parameter{\frontierMarginalCost}. That is, our example embodies substantial learning-by-aging: a firm's technology improves with the time it has spent in the market rather than with its cumulative output. Cournot competition gives firms' profits. These weakly increase with a firm's own technology type and weakly decrease with its rival's. Therefore, this example satisfies Assumptions \ref{ass:monotoneBoundedProfits}.\ref{ass:monotoneBoundedProfits:ownType} and \ref{ass:monotoneBoundedProfits}.\ref{ass:monotoneBoundedProfits:rivalType}. Since market size is bounded from above, it also satisfies Assumption \ref{ass:monotoneBoundedProfits}.\ref{ass:monotoneBoundedProfits:bounded}. 

Our baseline specification sets the sunk costs of entry to zero, so it satisfies Assumption \ref{ass:entryBarriers}. We set the process $\{Z_t\}$ equal to an \emph{i.i.d.} sequence of standard uniform random variables that is independent of $\{C_t\}$, satisfying Assumption \ref{ass:conditionalIndependence}. We set $\kappa(c,z)=\underline{\kappa}\left(1-z\right)^{-1/\gamma}$, which does not vary with $c$. This functional form gives $\kappa(c,Z)$ a Pareto distribution with minimum value $\underline{\kappa}$ and, when $\gamma>1$, mean $\underline{\kappa}\gamma/(\gamma-1)$. It also gives a firm's exit probability a constant elasticity with respect to its post-survival value. A monopolist with post-survival value $w\geq\underline{\kappa}$ exits with probability $\left(\underline{\kappa}/w\right)^{\gamma}$, so one percent more value brings $\gamma$ percent less exit risk. We set $\gamma=\counted{\exitElasticity}$. So that the median draw of the continuation cost equals one, we set $\underline{\kappa}=2^{-1/\gamma}$. That median is the unit in which we report every surplus below. Because $Z$ has a continuous distribution and $\kappa(c,z)$ strictly increases with $z$, $\kappa(c,Z)$ has no atom for every $c$; and because $\varphi$ equals zero, neither has $\kappa(c,Z)+\varphi(x,(c,Z))$. Assumption \ref{ass:continuousShocks} holds good.

Given these choices, we set the values of $b$ and $\omega$ to match two arbitrary-but-interpretable targets on the choices of a monopolist operating the best-possible technology at the median mature market size. Its Lerner index, the margin of price over marginal cost as a share of price, equals \chosen{\lernerTarget}; and its discounted monthly profit exactly equals the median draw of the continuation cost. Table \ref{tab:exampleTwoParameters} lists the values of these and all other parameters.

\begin{table}[tp]\centering
\begin{tabular}{llr}
\hline
Symbol & Description & Value \\
\hline
--- & Annual Cost of Capital, Percent & \counted{\annualCostOfCapital} \\
$\beta$ & Monthly Discount Factor & \parameter{\discountFactor} \\
--- & Expected Months of Demand Immaturity & \counted{\developmentMonths} \\[4pt]
$\bar{p}$ & Choke Price & \counted{\demandIntercept} \\
--- & Growth Factor of the Number of Customers & \chosen{\growthStep} \\
--- & Probability that the Number of Customers Grows & \probability{\growthProbability} \\
--- & Points on the Market Size Grid & \counted{\marketSizes} \\[4pt]
$\check{k}$ & Technology Types & \counted{\technologyTypes} \\
$\log\mu_{k-1}-\log\mu_k$ & Log Cost Gap between Types & \chosen{\costDecline} \\
--- & Probability of Advancing a Type & \probability{\advanceProbability} \\
$\mu_1$ & Entrant's Marginal Cost & \counted{\entrantMarginalCost} \\
$\mu_{\check{k}}$ & Frontier Marginal Cost & \parameter{\frontierMarginalCost} \\[4pt]
$\gamma$ & Pareto Index of the Continuation Cost's Distribution & \counted{\exitElasticity} \\
$\underline{\kappa}$ & Lowest Continuation Cost & \parameter{\lowestContinuationCost} \\
--- & Median Continuation Cost & \counted{\medianContinuationCost} \\[4pt]
$b$ & Demand Slope & \parameter{\demandSlope} \\
$\omega$ & Cost Curvature & \parameter{\costCurvature} \\[4pt]
--- & Convergence Tolerance & \powerOfTwo{\valueToleranceExponent} \\
--- & Months Simulated per Market & \counted{\simulationMonths} \\
\hline
\end{tabular}

\caption{The Example's Parameter Values}
\label{tab:exampleTwoParameters}
\end{table}

\subsection{Social Planning versus Laissez Faire}

The choices of a social planner who controls entry, exit, and production decisions of all firms give our point of reference for competition policy. This planner balances the benefits of entering new markets early, so that firms can exploit learning-by-aging during the market's immaturity, against the costs of over-investing in markets that stop growing while still too small to justify ongoing production.

The planner begins each period with zero, one, or two firms, of types $h$ and $l$ with $h\geq l$. A superscript $\varsigma$ (for $\varsigma$ocial) denotes the planner's counterpart of an object in the duopoly game. The planner first makes every static choice and generates the surplus $\pi^\varsigma(h,l,c)$. For the current example, $\pi^\varsigma(h,l,c)$ is total money-metric utility less the cost of producing it, with output allocated across the active firms so as to equate marginal costs among those producing positive output. After receiving the surplus, the planner decides on entry into each empty slot, and on continuation for each active firm. Entry costs the planner $\varphi^\varsigma(y)$, and continuation costs her $\kappa^\varsigma(y)$. These satisfy $\varphi^\varsigma(y)\leq\varphi(x,y)$ for every $x$ and $\kappa^\varsigma(y)\leq\kappa(y)$. Allowing firms' costs and the planner's costs to differ lets us account for entry and continuation taxes. In our baseline experiments, the planner's entry and continuation costs take the calibrated values in Table \ref{tab:exampleTwoParameters}. However, we experiment with entry fees which drive $\varphi(x,y)>\varphi^\varsigma(y)$.   The period ends with the stochastic evolution of any active firms' types and of the demand state. The planner discounts future surpluses with $\beta$, the discount factor the firms use. Specifying this problem as a dynamic program and solving it with Bellman equation iteration are both standard exercises.

Given the planner's solution, we calculate the expected surplus generated in a new market with no incumbent firms. This and every other expected surplus we report comes from value functions computed to the tolerance in Table \ref{tab:exampleTwoParameters}. Since our numeraire is a flow (the median continuation cost), we report the constant monthly payment with the same present value. This equals \surplus{\plannerSurplus}. The analogous value for the duopoly model's equilibrium equals only \surplus{\laissezFaireSurplus}. The difference between these two expected surpluses equals the potential for competition policy to improve outcomes. 

%
%

\begin{figure}[tp]
\centering
\begin{tikzpicture}
\begin{axis}[
  exampleTwoAxis,
  xmode=log,
  height=0.38\textheight,
  xmin=1, xmax=\largestSizeDrawn,
  xtick={1,\matureSizeMedian, \matureSizeSeventyFifth,
    \matureSizeNinetieth, \matureSizeNinetyFifth,
    \matureSizeNinetyEighth, \matureSizeNinetyNinth, \largestSizeDrawn},
  xticklabels={,\marketsize{\matureSizeMedian},
    \marketsize{\matureSizeSeventyFifth},
    \marketsize{\matureSizeNinetieth},
    \marketsize{\matureSizeNinetyFifth},
    \marketsize{\matureSizeNinetyEighth},
    \marketsize{\matureSizeNinetyNinth}, },
  minor xtick={},
  ymin=-0.012, ymax=0.18,
  ytick={0, 0.04, 0.08, 0.12, 0.16},
  xlabel={Market Size (Log Scale)},
  ylabel={Surplus},
  legend pos=north west,
]

\addplot[plannerLine] table[x=size, y=planner]
  {compute/acy1a/accumulatedSurplus.dat};
\addlegendentry{Planner}
\addplot[marketLine] table[x=size, y=freeEntry]
  {compute/acy1a/accumulatedSurplus.dat};
\addlegendentry{Laissez Faire}

\addplot[plannerLine]
  coordinates {(26, \plannerSurplus) (30, \plannerSurplus)};
\addplot[marketLine]
  coordinates {(26, \laissezFaireSurplus) (30, \laissezFaireSurplus)};

%
\node[anchor=east, font=\small, inner sep=2pt] (plannerTotal)
  at (axis cs:23, 0.168) {\surplus{\plannerSurplus}};
\draw[->, shorten >=2pt] (plannerTotal.south east) -- (axis cs:26.6, 0.1525);
\node[anchor=east, font=\small, inner sep=2pt] (marketTotal)
  at (axis cs:23, 0.098) {\surplus{\laissezFaireSurplus}};
\draw[->, shorten >=2pt] (marketTotal.south east) -- (axis cs:26.6, 0.0825);

\end{axis}
\end{tikzpicture}
\caption{Accumulated Surplus}
\label{fig:exampleTwoAccumulatedSurplus}

\vspace*{3pt}
\figurenote{\noindent Note: Each point on a curve gives the part of a new market's expected surplus that comes from markets maturing at or below the $x$-axis value. Markets that mature at a given size add their expected surplus times the probability that a market matures at that size. The two short lines to the right mark the total surpluses for the planner and laissez faire solutions. The labeled horizontal tick marks report the 50th, 75th, 90th, 95th, 98th, and 99th percentiles of market size at maturity. The two unlabeled tick marks are at $1$, the initial and smallest size, and $30$. The $x$-axis does \emph{not} cover the entire support for mature market size.  The average surpluses by mature market size equal conditional averages of realized surplus from \counted{\marketsSimulated} randomly sampled markets simulated for \counted{\simulationMonths} months each.

Months after the last one simulated are discounted by at least $\beta^{\simulationMonths}=\scientific{\discountAtHorizon}$. A month's surplus can be negative, because firms pay continuation costs. No firm pays more to continue than continuing is worth, and nothing is worth more than two frontier firms serving the largest market forever. Cutting each simulation at \counted{\simulationMonths} months therefore moves no reported surplus by more than \scientific{\truncationBound}. The standard error of every plotted point is at most \parameter{\largestAccumulatedStandardError}.}

\end{figure}
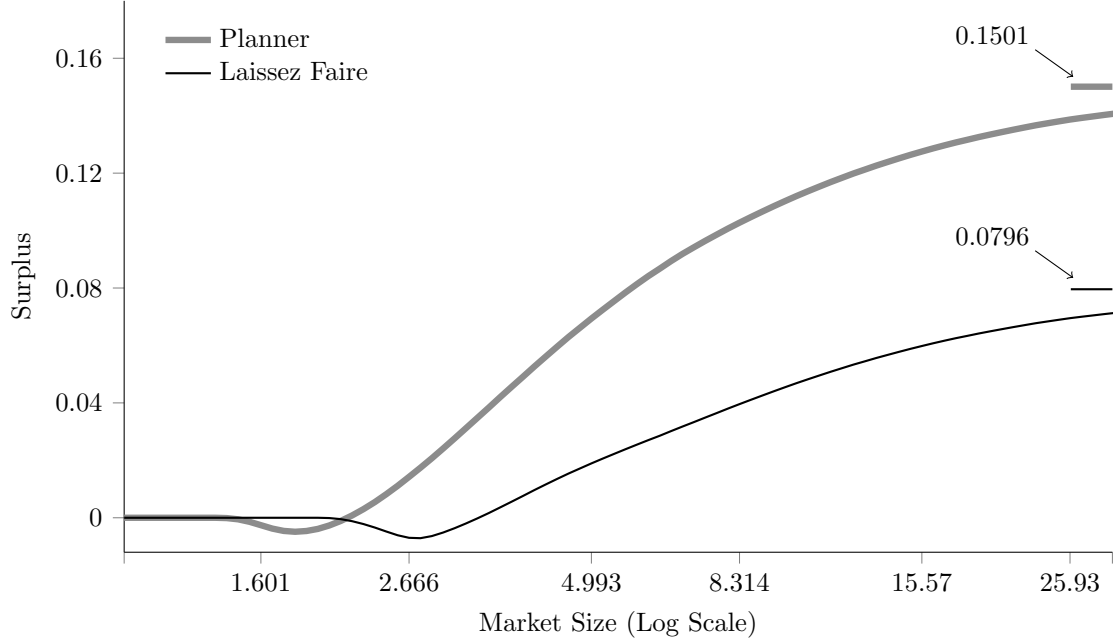

Before considering regulatory responses to this market failure, we wish to identify where and how the planner improves on laissez faire policy. We partially answer ``where'' by calculating the contributions of expected surplus made by markets with different mature sizes. Figure \ref{fig:exampleTwoAccumulatedSurplus} presents the results. Market size at maturity runs along its horizontal axis, while its vertical axis gives surplus. The two curves give cumulative expected surplus for the planner and for laissez faire. Each point on a curve gives the expected surplus received from markets with mature market size at or below the given $x$-axis value. A curve decreases whenever markets have negative expected surplus. We limit the $x$-axis to mature market sizes between $1$ (the initial and smallest value) and $\counted{\largestSizeDrawn}$, which exceeds the 99th percentile of mature market size. Limiting the $x$-axis range by omitting the very largest markets makes the figure's relevant content more visible. Tick marks indicate the position of the mature market size's median ($\marketsize{\matureSizeMedian}$) and its 75th, 90th, 95th, 98th, and 99th percentiles ($\marketsize{\matureSizeSeventyFifth}$, $\marketsize{\matureSizeNinetieth}$, $\marketsize{\matureSizeNinetyFifth}$, $\marketsize{\matureSizeNinetyEighth}$, and $\marketsize{\matureSizeNinetyNinth}$). If the $x$-axis covered the entire support of the mature market size, the curves' rightmost values would equal the surpluses reported above. Short lines on the right-hand side mark these values. 

For both the planner and laissez faire, the majority of markets generate no positive surplus. Both curves fall below zero before they rise above it. The planner's curve falls first, over markets that grew too little to pay the continuation costs of the firms she created. Laissez faire's curve falls over somewhat larger markets. Above the \parameter{\gapOpensPercentile}th percentile of mature market size ($\marketsize{\gapOpensSize}$), the planner's curve exceeds laissez faire's. At the 99th percentile of mature market size, the gap between them has nearly reached \surplus{\surplusGap}, the difference between the two total expected surpluses. Markets above that size therefore account for all of the planner's net gain over laissez faire.

%
%

\begin{figure}[tp]
\centering
\begin{tikzpicture}
\begin{groupplot}[
  exampleTwoAxis,
  xmode=log,
  group style={
    group size=1 by 2,
    vertical sep=1.6cm,
    x descriptions at=edge bottom,
  },
  xmin=1, xmax=\largestSizeDrawn,
  xtick={1, \matureSizeMedian, \matureSizeSeventyFifth,
    \matureSizeNinetieth, \matureSizeNinetyFifth,
    \matureSizeNinetyEighth, \matureSizeNinetyNinth, \largestSizeDrawn},
  xticklabels={, \marketsize{\matureSizeMedian},
    \marketsize{\matureSizeSeventyFifth},
    \marketsize{\matureSizeNinetieth},
    \marketsize{\matureSizeNinetyFifth},
    \marketsize{\matureSizeNinetyEighth},
    \marketsize{\matureSizeNinetyNinth}, },
  minor xtick={},
  ymin=0, ymax=1.04,
  ytick={0, 0.2, 0.4, 0.6, 0.8, 1},
  ylabel={Frequency},
]

\nextgroupplot[title={One Firm}]
\addplot[plannerLine] table[x=size, y=onePlannerOnArrival]
  {compute/acy1a/structureOnArrival.dat};
\addlegendentry{Planner, on arrival}
\addplot[plannerSettled] table[x=size, y=onePlannerAfterMaturity]
  {compute/acy1a/firmCountsByMatureSize.dat};
\addlegendentry{Planner, $t_M+36$}
\addplot[marketLine] table[x=size, y=oneFreeEntryOnArrival]
  {compute/acy1a/structureOnArrival.dat};
\addlegendentry{Laissez Faire, on arrival}
\addplot[marketSettled] table[x=size, y=oneFreeEntryAfterMaturity]
  {compute/acy1a/firmCountsByMatureSize.dat};
\addlegendentry{Laissez Faire, $t_M+36$}

\nextgroupplot[title={Two Firms}, xlabel={Market Size (Log Scale)}]
\addplot[plannerLine] table[x=size, y=twoPlannerOnArrival]
  {compute/acy1a/structureOnArrival.dat};
\addplot[plannerSettled] table[x=size, y=twoPlannerAfterMaturity]
  {compute/acy1a/firmCountsByMatureSize.dat};
\addplot[marketLine] table[x=size, y=twoFreeEntryOnArrival]
  {compute/acy1a/structureOnArrival.dat};
\addplot[marketSettled] table[x=size, y=twoFreeEntryAfterMaturity]
  {compute/acy1a/firmCountsByMatureSize.dat};

\end{groupplot}
\end{tikzpicture}
\caption{Market Structure on Arrival and 36 Months after Maturation}
\label{fig:exampleTwoFirmCounts}

\vspace*{3pt}
\figurenote{\noindent Note: The top panel reports frequencies of a market having one active firm by market size. The two solid lines' frequencies are tabulated as of the moment the market reaches the given market size (on arrival), while the dashed lines' frequencies are tabulated 36 months after the market matured at that size ($t_M+36$). The bottom panel reports the analogous frequencies of a market having two active firms. The horizontal axis tick marks are identical to those in Figure \ref{fig:exampleTwoAccumulatedSurplus}. The standard error of every plotted point is at most \parameter{\largestFirmCountStandardError}.}

\end{figure}
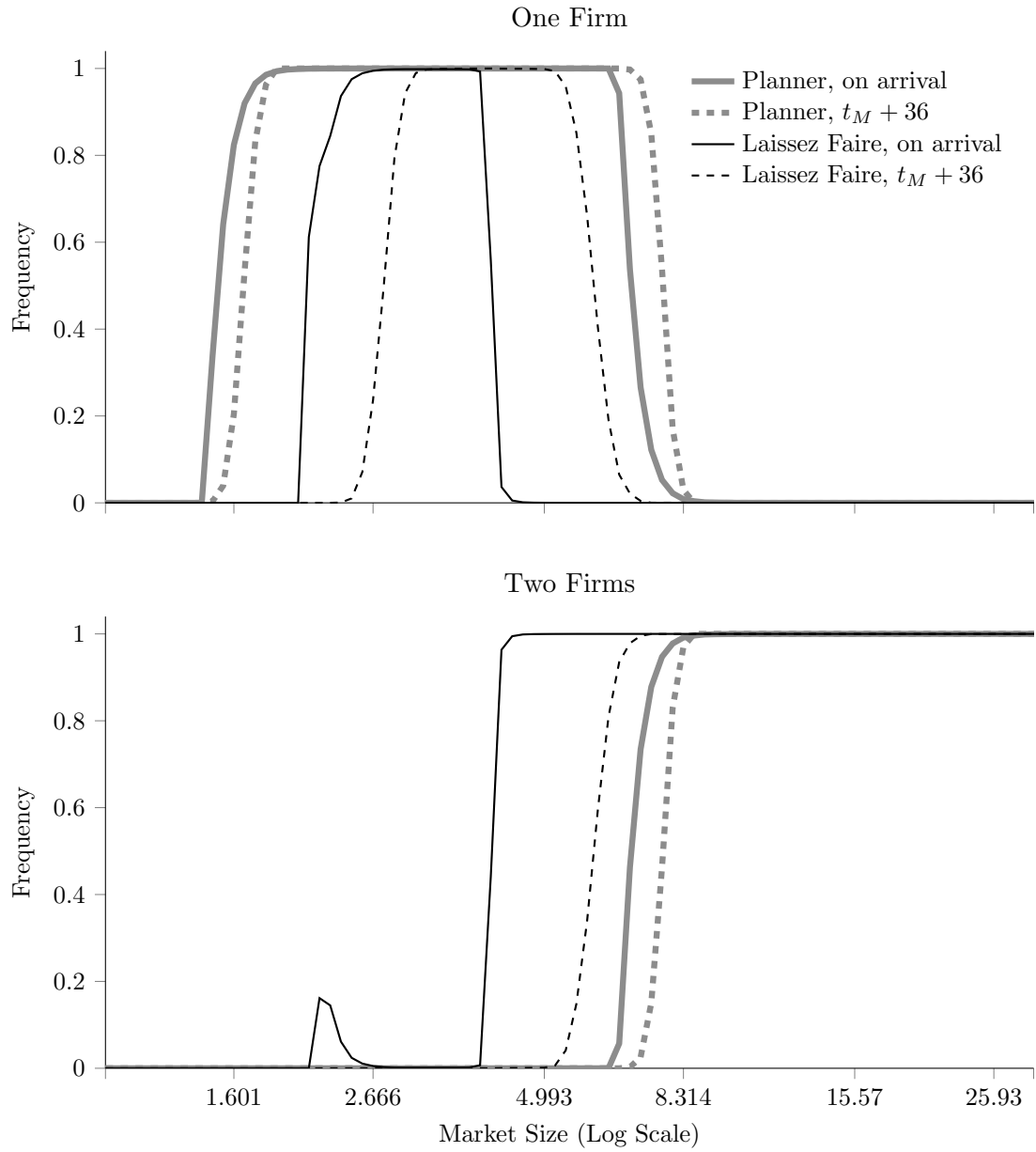

Figure \ref{fig:exampleTwoFirmCounts} partially answers ``how'' by plotting how often a market has one (top panel) or two (bottom panel) active firms, by market size. Recall that each market's size can only increase. The gray solid lines give the frequencies for the planner when the market size reaches the given value on the horizontal axis. The gray dashed lines give the analogous frequencies by market size \counted{\monthsAfterMaturity} months after the market matured at that size. Whether a market matures at a size is independent of its history, so the markets that reach a size and those that mature there hold the same distribution of firms on arrival. The solid and dashed lines therefore describe the same markets, in distribution, at two moments. The solid-black and dashed-black lines give the same frequencies for the laissez faire solution. 

Inspection of this figure leaves us with three lessons. First, the planner's first firm enters well before the first firm under laissez faire. Second, the planner requires much larger market sizes to induce a duopolist's entry than does laissez faire. Third, markets often experience a ``shakeout'' after demand reaches maturity. The planner's shakeout is much smaller than that under laissez faire. Although the planner takes on more demand risk, fewer of her firms exit after demand matures.

%
%
%

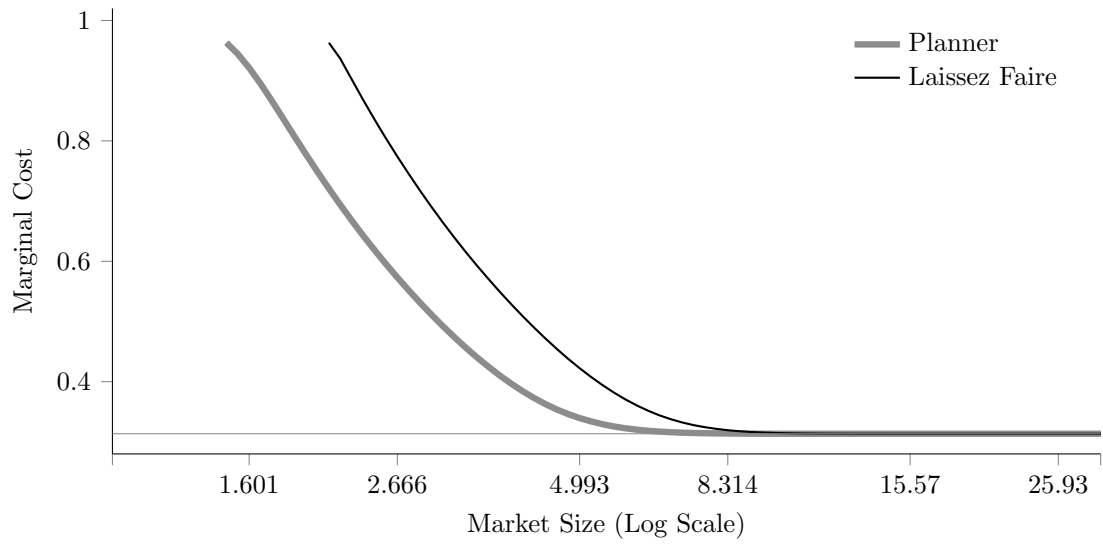
\begin{figure}[tp]
\centering
\begin{tikzpicture}
\begin{axis}[
  exampleTwoAxis,
  xmode=log,
  height=0.32\textheight,
  xmin=1, xmax=\largestSizeDrawn,
  xtick={1, \matureSizeMedian, \matureSizeSeventyFifth,
    \matureSizeNinetieth, \matureSizeNinetyFifth,
    \matureSizeNinetyEighth, \matureSizeNinetyNinth, \largestSizeDrawn},
  xticklabels={, \marketsize{\matureSizeMedian},
    \marketsize{\matureSizeSeventyFifth},
    \marketsize{\matureSizeNinetieth},
    \marketsize{\matureSizeNinetyFifth},
    \marketsize{\matureSizeNinetyEighth},
    \marketsize{\matureSizeNinetyNinth}, },
  minor xtick={},
  ymin=0.28, ymax=1.02,
  ytick={0.4, 0.6, 0.8, 1.0},
  xlabel={Market Size (Log Scale)},
  ylabel={Marginal Cost},
  legend pos=north east,
]

\addplot[plannerLine] table[x=size, y=bestMarginalCostPlannerOnArrival]
  {compute/acy1a/structureOnArrival.dat};
\addlegendentry{Planner}
\addplot[marketLine] table[x=size, y=bestMarginalCostFreeEntryOnArrival]
  {compute/acy1a/structureOnArrival.dat};
\addlegendentry{Laissez Faire}

\addplot[plannerGray, line width=0.4pt, mark=none, forget plot]
  coordinates {(1, \frontierMarginalCost) (\largestSizeDrawn, \frontierMarginalCost)};

\end{axis}
\end{tikzpicture}
\caption{The Leading Firm's Marginal Cost on Arrival}
\label{fig:exampleTwoLeaderCost}

\vspace*{3pt}
\figurenote{\noindent Note: Each curve reports the average marginal cost at zero
output of the most productive active firm in the market, as of the moment the
market reaches the given market size (on arrival). The average covers
only those markets holding at least one firm, so the laissez faire curve
begins where its first firms do. The thin horizontal line marks the
marginal cost of the frontier technology, \parameter{\frontierMarginalCost}.
The horizontal axis tick marks are identical to those in Figure
\ref{fig:exampleTwoAccumulatedSurplus}. The standard error of every plotted point is at most \parameter{\largestLeaderCostStandardError}.}

\end{figure}

The planner enters before demand is mature to reduce the costs of waiting for technological improvements. Figure \ref{fig:exampleTwoLeaderCost} displays these benefits by plotting the average (across markets) of the lowest marginal cost (across firms) evaluated at zero output. The sampling occurs at the moment the market attains that size, as with the solid lines in Figure \ref{fig:exampleTwoFirmCounts}.  Because the planner enters markets earlier, the planner's curve begins to the left of its laissez faire counterpart. Both curves attain the smallest possible marginal cost in the largest markets, but the planner's costs are substantially below those of private firms in mid-sized markets. 

Our comparison of planning and laissez faire outcomes has focused on differences in market structure: under what circumstances do firms enter and exit? We have paid relatively little attention to conduct and (static) performance. To learn how much of the surplus reduction would disappear if firms priced efficiently, we calculate the equilibrium of our model in which a regulator acts exactly as a Walrasian auctioneer. Each month, the regulator begins knowing both the number of customers and every firm's costs. The regulator calculates a market-clearing price, offers to purchase from the firms as many widgets as they desire to sell, and sells them at that price to the customers. This removes market power in pricing, but the firms still act strategically when choosing entry and exit. The average realized surplus after adopting this ideal product-market regulation is slightly \emph{less} than the laissez faire surplus. Efficient pricing therefore recovers none of the lost surplus, and it loses a little more. The loss lies in firms' entry and exit choices.

\subsection{Regulation with Dynamic Incentives}

Consider a two-stage model of an initially empty market, in which firms enter one at a time and then produce and sell. There, omniscient regulation achieves the planner's outcome by pricing widgets at their marginal cost and paying each producer the surplus its presence generates. We take this approach as our baseline scheme for the regulation of dynamic incentives. As in our earlier experiment, the regulator acts as a Walrasian auctioneer, so each producer keeps the market-clearing price on every unit it sells, less its costs. She also pays a firm the difference its presence makes to the surplus of everyone else in the market: its customers and its rival, if it has one. Producers make their own entry and exit decisions. Our code verifies that each regime's profits satisfy Assumption \ref{ass:monotoneBoundedProfits}. For our example, the average surplus gained under this scheme equals only \surplus{\regulationSurplus} on a flow basis. That is, \emph{optimal product-market regulation combined with the optimal dynamic incentives from a two-stage model underperforms laissez faire.} 

This scheme pays each firm, in every month, exactly its contribution to that month's total surplus. However, a firm's payment leaves out how its presence changes other firms' later entry and exit decisions. The planner takes these dynamic effects into account. In principle, we could use the planner's value function to construct state-contingent entry and exit fees that induce entry and most exit decisions to match her own. However, we choose to close this section with a cruder but quantitatively relevant experiment: we impose an entry fee on all new firms. The proceeds get rebated back to the market's customers as lump-sum transfers. The regulator continues with the optimal regulation for a two-stage model, with Walrasian pricing and these transfers. 

%

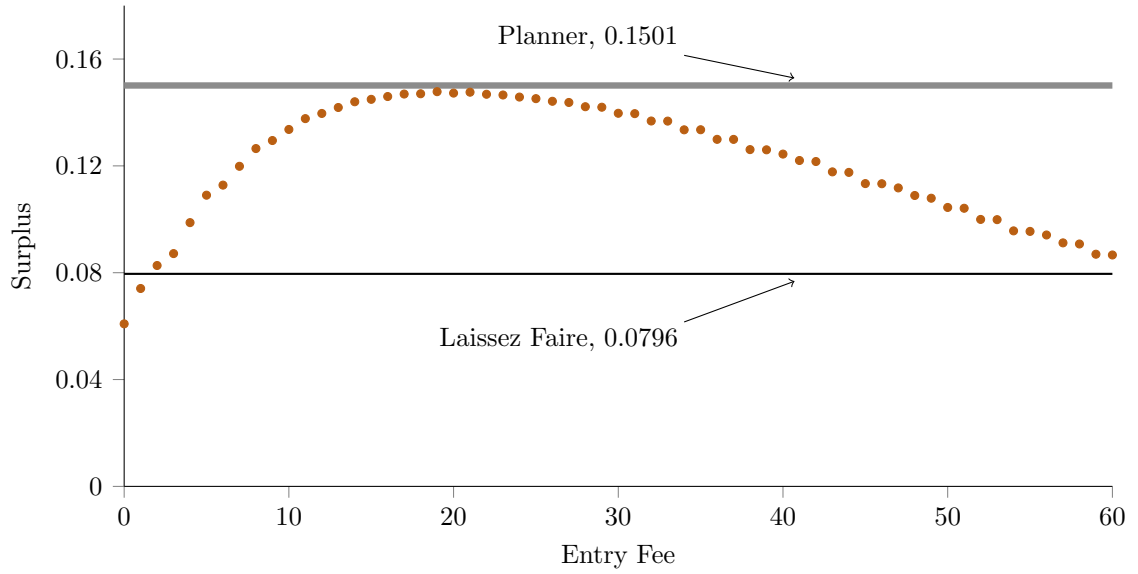
\begin{figure}[tp]
\centering
\begin{tikzpicture}
\begin{axis}[
  exampleTwoAxis,
  height=0.34\textheight,
  xmin=0, xmax=60,
  xtick={0, 10, 20, 30, 40, 50, 60},
  ymin=0, ymax=0.18,
  ytick={0, 0.04, 0.08, 0.12, 0.16},
  xlabel={Entry Fee},
  ylabel={Surplus},
]

\addplot[plannerLine, forget plot]
  coordinates {(0, \plannerSurplus) (60, \plannerSurplus)};
\addplot[marketLine, forget plot]
  coordinates {(0, \laissezFaireSurplus) (60, \laissezFaireSurplus)};

\addplot[solvedFeePoint]
  table[x=entryFee, y=flow]
  {compute/acy1a/surplusByEntryFee.dat};

\node[anchor=east, font=\small, inner sep=2pt] (plannerLevel)
  at (axis cs:34, 0.168) {Planner, \surplus{\plannerSurplus}};
\draw[->, shorten >=2pt] (plannerLevel.south east) -- (axis cs:41, 0.1525);
\node[anchor=east, font=\small, inner sep=2pt] (laissezLevel)
  at (axis cs:34, 0.055) {Laissez Faire, \surplus{\laissezFaireSurplus}};
\draw[->, shorten >=2pt] (laissezLevel.north east) -- (axis cs:41, 0.0775);

\end{axis}
\end{tikzpicture}
\caption{Average Surplus with Constant Entry Fees}
\label{fig:exampleTwoSurplusByEntryFee}

\vspace*{3pt}
\figurenote{\noindent Note: Each circle reports the expected surplus of a new
market when every entrant pays the fee on the horizontal axis. In all experiments, the regulator acts as a Walrasian auctioneer. Firms
keep the revenue their sales earn at the market-clearing price. Each month the
regulator also pays each of them the difference its presence makes to the
surplus of its customers and its rival.}
\end{figure}

Figure \ref{fig:exampleTwoSurplusByEntryFee} plots the results for our entry fee experiments. We vary the fee from zero to \counted{\largestEntryFee}. The welfare at zero equals the surplus from our motivating experiment, \surplus{\regulationSurplus}. The figure demonstrates that imposing a constant fee on a regulated product market might increase total surplus. Indeed, total surplus exceeds that produced under laissez faire at every fee we examine from \counted{\lowestBeatingFee} through \counted{\highestBeatingFee}. The highest surplus achieved on the grid of entry fees we examine equals \surplus{\bestEntryFeeSurplus}, which an entry fee of \counted{\bestEntryFee} achieves.

\section{Conclusion}
\label{sec:conclusion}

Because it studies only one example, our quantitative analysis of competition policy demonstrates possibility rather than empirical relevance. Nevertheless, the large welfare gains a constant entry fee delivered lead us to ask how widely this result applies to other economic environments. Thoroughly answering this question will require calculating thousands of counterfactual equilibria for hundreds of specifications. For Section \ref{sec:competitionPolicy}'s entry fee experiments, we computed \counted{\equilibriaTimed} equilibria with Algorithm \ref{alg:candidateEquilibrium}, one for each fee. We implemented the algorithm in MATLAB\textsuperscript{\textregistered} and ran it on one computer. Its operating system identifies its \ifnum\solvingSockets=1 processor as an \solvingProcessor\ with \counted{\solvingCoresPerSocket} cores\else\counted{\solvingSockets} processors as \solvingProcessor, with \counted{\solvingCoresPerSocket} cores each\fi. Computing one of these equilibria took \seconds{\secondsPerEquilibrium} seconds on one core, on average. Valuing it, which our welfare comparisons require, took a further \seconds{\secondsPerValuation} seconds. At the first rate, computing one hundred thousand equilibria would occupy this computer's \counted{\poolWorkers} cores for \parameter{\projectedDays} days. This extrapolation of our experience convinces us that our algorithm can carry out this more extensive analysis on a modest budget.

\clearpage
\newpage

\pdfbookmark[0]{References}{pdfbm:refs}
\bibliographystyle{ecta}
\bibliography{alljaap2}

\clearpage

\appendix
\section*{Appendix: Proofs of the Lemmas in Section \ref{sec:twoTypes}}

\begin{proof}[Proof of Lemma \ref{lem:twoTypesHighMonopolyBeatsDuopoly}]
Compare a type $2$ incumbent monopolist with a type $2$ incumbent duopolist facing a type $2$ rival, immediately after exit decisions in period $t_0$, each facing $C_{t_0}=c$ and sharing the same exogenous process $Y_t$. Let $\tau$ be the first period $t\geq t_0+1$ in which $a_S(2,2,Y_t)<1$, with $\tau=\infty$ if this never occurs. Before $\tau$, both duopolists survive for sure, and both keep type $2$, the highest type, because types cannot regress (Assumption \ref{ass:noRegress}). In equilibrium, the incumbent duopolist therefore earns $\pi(2,2,Y_t)$ in periods $t_0+1\leq t\leq\tau$ and pays $\kappa(Y_t)$ in periods $t_0+1\leq t<\tau$. If $\tau<\infty$, the duopolist exits with positive probability in period $\tau$, so exit is a best response. Its expected payoff from that period's survival stage onward therefore equals zero. The current expected discounted value of these profits less continuation costs equals $v_S(2,2,c)$.

The incumbent monopolist may deviate from its equilibrium strategy by surviving in periods $t<\tau$ and exiting in period $\tau$ if $\tau<\infty$. While active, that firm faces no rival, a type $1$ rival, or a type $2$ rival. Assumption \ref{ass:monotoneBoundedProfits}.\ref{ass:monotoneBoundedProfits:rivalType} therefore gives it profits weakly greater than $\pi(2,2,Y_t)$ through period $\tau$, and it pays the same continuation costs $\kappa(Y_t)$ strictly before $\tau$. It earns a continuation value of zero upon exit. The expected discounted payoff from this feasible action sequence therefore weakly exceeds $v_S(2,2,c)$. Because $v_S(2,0,c)$ equals the highest possible payoff to the incumbent monopolist, it weakly exceeds the payoff from this feasible action sequence. Therefore, $v_S(2,2,c)\leq v_S(2,0,c)$.
\end{proof}

\begin{proof}[Proof of Lemma \ref{lem:twoTypesLowSurvivalRule}]
Suppose first that $a_S(2,1,y) v_S(1,2,c)+(1-a_S(2,1,y)) v_S(1,0,c)>\kappa(y)$. Profit maximization then gives $a_S(1,2,y)=1$. Because the equilibrium is natural, $a_S(2,1,y)=1$. The left-hand side therefore equals $v_S(1,2,c)$, and $v_S(1,2,c)>\kappa(y)$. 

Next, assume that $v_S(1,2,c)>\kappa(y)$. We show that $v_S(1,0,c)\geq v_S(1,2,c)$. The argument follows the proof of Lemma \ref{lem:twoTypesHighMonopolyBeatsDuopoly}, with one addition: the type $1$ firm's own type may advance. Compare a type $1$ incumbent facing a type $2$ rival with a type $1$ incumbent monopolist, each starting immediately after exit decisions in period $t_0$ with $C_{t_0}=c$ and sharing the same processes $Y_t$ and $K_t$, with $K_{t_0}=1$. Let $\tau$ be the first period $t\geq t_0+1$ in which the duopolist might exit, $a_S(K_t,2,Y_t)<1$, with $\tau=\infty$ if this never occurs.

Before $\tau$, the duopolist survives for sure, and so does its rival. The rival keeps type $2$, because types cannot regress (Assumption \ref{ass:noRegress}). While the duopolist has type $1$, its rival survives because the equilibrium is natural: a type $2$ firm survives for sure whenever its type $1$ rival does. Once the duopolist has type $2$, its rival follows the same survival rule, $a_S(2,2,Y_t)=1$. The duopolist therefore earns $\pi(K_t,2,Y_t)$ in periods $t_0+1\leq t\leq\tau$ and pays $\kappa(Y_t)$ in periods $t_0+1\leq t<\tau$. As in the proof of Lemma \ref{lem:twoTypesHighMonopolyBeatsDuopoly}, its expected payoff from period $\tau$'s survival stage onward equals zero. Consequently, $v_S(1,2,c)$ equals the expected discounted value of these profits less continuation costs.

The monopolist may deviate from its equilibrium strategy by surviving in periods $t<\tau$ and exiting in period $\tau$ if $\tau<\infty$. This decision depends only on its own type $K_t$ and the exogenous state $Y_t$. While doing so it faces no rival or a rival of type at most $2$, so Assumption \ref{ass:monotoneBoundedProfits}.\ref{ass:monotoneBoundedProfits:rivalType} gives it a flow payoff weakly greater than $\pi(K_t,2,Y_t)$, and it pays the same $\kappa(Y_t)$ strictly before $\tau$. It earns zero upon exit. Because this action sequence is feasible, $v_S(1,0,c)$ weakly exceeds its payoff, which weakly exceeds $v_S(1,2,c)$. Combining $v_S(1,0,c)\geq v_S(1,2,c)$ with $v_S(1,2,c)>\kappa(y)$ yields $a_S(2,1,y) v_S(1,2,c)+(1-a_S(2,1,y)) v_S(1,0,c)>\kappa(y)$.
\end{proof}

\begin{proof}[Proof of Lemma \ref{lem:twoTypesCalculatedHighValueOrder}]
Because $\pi(2,0,y)\geq \pi(2,1,y) \geq \pi(2,2,y)$, the Bellman operator defined by the right-hand sides of \eqref{eq:twoTypesBellmanHighMonopoly} and \eqref{eq:twoTypesBellmanHighLow} maps the set of pairs of functions $f_0$ and $f_1$ satisfying $f_0(c)\geq f_1(c)\geq w_S(2,2,c)$ for every $c\in\mathcal{C}$ into itself.  Therefore, the operator's fixed point satisfies the first two inequalities. For the third, $\pi(1,2,y)\leq\pi(2,2,y)$, so comparing \eqref{eq:twoTypesBellmanHighHigh} and \eqref{eq:twoTypesBellmanLowHigh} shows that the operator defined by the right-hand side of \eqref{eq:twoTypesBellmanLowHigh} maps the set of functions $f$ satisfying $f(c)\leq w_S(2,2,c)$ for every $c\in\mathcal{C}$ into itself. Therefore, its fixed point satisfies $w_S(1,2,c)\leq w_S(2,2,c)$.
\end{proof}

\begin{proof}[Proof of Lemma \ref{lem:twoTypesLowMonopolyBeatsDuopoly}]
As in the proof of Lemma \ref{lem:twoTypesHighMonopolyBeatsDuopoly}, we give the monopolist a feasible plan that earns at least the duopolist's equilibrium payoff. Two things are new here: the firm's type may advance, and entry may give the monopolist a rival.

Compare a type $1$ incumbent monopolist with a type $1$ incumbent duopolist, each starting immediately after exit decisions in period $t_0$ with $C_{t_0}=c$ and sharing the same processes $Y_t$ and $K_t$, with $K_{t_0}=1$. Denote the (random) rival types faced by the monopolist and duopolist when period-$t$ profits accrue with $X_t^M$ and $X_t^D$. Let $\tau$ be the first period $t\geq t_0+1$ in which either (i) $K_t=2$ or (ii) $K_t=1$ and $a_S(1,X_t^D,Y_t)<1$, with $\tau=\infty$ if neither occurs. Before $\tau$, the duopolist survives for sure. Its rival also survives for sure: by symmetry when $X_t^D=1$, and because the equilibrium is natural when $X_t^D=2$. The duopolist's value $v_S(1,1,c)$ therefore equals the expected discounted value of its profits through period $\tau$, less its continuation costs before $\tau$, plus its payoff from period $\tau$'s survival stage and thereafter. The rivals' types start in order: $X_{t_0}^M=0<X_{t_0}^D=1$. Couple the two subgames so that this order persists, with $X_t^M\leq X_t^D$ through period $\tau$.\footnote{The coupling places the two subgames on a common probability space and draws the duopolist's rival's technology transitions so that its type never falls below that of any rival to the initial monopolist. Entry and exit cannot reverse the order, because an entrant has type $1$ and the duopolist's rival does not exit before $\tau$. The coupling is constructed sequentially: a deviating firm can simulate the duopolist's history from its own observed history and its private device, without using future information. The proof of Theorem \ref{th:uniqueness} in Online Supplement S2 \citep*{duopolysupp:abbringetal} presents such a construction in detail.} Applying Assumption \ref{ass:monotoneBoundedProfits} yields $\pi(1,X_t^M,Y_t)\geq \pi(1,X_t^D,Y_t)$ before $\tau$ and, if $\tau<\infty$, $\pi(K_\tau,X_\tau^M,Y_\tau)\geq\pi(K_\tau,X_\tau^D,Y_\tau)$ in period $\tau$.

The incumbent monopolist may deviate from its equilibrium strategy as follows. Before $\tau$, it survives for sure. If $\tau<\infty$ and $K_\tau=1$, it exits in period $\tau$. If $\tau<\infty$ and $K_\tau=2$, it follows its equilibrium strategy from the survival decision in period $\tau$ and thereafter.

Before $\tau$, the deviating monopolist earns weakly greater profits than the duopolist, and both pay the same continuation costs.

If $\tau<\infty$, let $\widehat X_\tau^M$ denote the type of any rival to the initial monopolist after period $\tau$'s entry stage. If $X_\tau^M>0$, then $\widehat X_\tau^M=X_\tau^M$; otherwise, $\widehat X_\tau^M$ equals either $0$ or $1$, depending on entry. The duopolist's market remains a duopoly through this stage, so $\widehat X_\tau^M\leq X_\tau^D$.

If $K_\tau=1$, the deviating monopolist exits and receives zero. Because $a_S(1,X_\tau^D,Y_\tau)<1$, the duopolist's equilibrium payoff from period $\tau$'s survival stage and thereafter also equals zero.

If $K_\tau=2$, the monopolist's expected payoff from that survival stage and thereafter is $v_E(2,\widehat X_\tau^M,Y_\tau)$, and the duopolist's is $v_E(2,X_\tau^D,Y_\tau)$. Steps 1 and 3 showed that $v_S(2,x,c)=w_S(2,x,c)$ for $x\in\{0,1,2\}$, so Lemma \ref{lem:twoTypesCalculatedHighValueOrder} gives $v_S(2,0,c)\geq v_S(2,1,c)\geq v_S(2,2,c)$. With $a=a_S(1,2,y)$, the post-entry values therefore satisfy
\begin{align*}
v_E(2,0,y)&=\max\{0,v_S(2,0,c)-\kappa(y)\}\\
&\geq\max\{0,a\,v_S(2,1,c)+(1-a)\,v_S(2,0,c)-\kappa(y)\}=v_E(2,1,y)\\
&\geq\max\{0,v_S(2,2,c)-\kappa(y)\}=v_E(2,2,y).
\end{align*}
Since $\widehat X_\tau^M\leq X_\tau^D$, $v_E(2,\widehat X_\tau^M,Y_\tau)\geq v_E(2,X_\tau^D,Y_\tau)$.

In both cases, the deviating monopolist's payoff from period $\tau$'s survival stage and thereafter weakly exceeds the duopolist's. Together with the profit comparison through period $\tau$ and equal continuation costs before $\tau$, the plan's expected discounted payoff weakly exceeds $v_S(1,1,c)$. No strategy earns the monopolist more than $v_S(1,0,c)$, so $v_S(1,0,c)\geq v_S(1,1,c)$.
\end{proof}

\begin{proof}[Proof of Lemma \ref{lem:twoTypesCandidateNatural}]
Suppose that $\alpha_S(1,2,y)>0$. Its construction in Step 2 implies that $w_S(1,2,c)>\kappa(y)$. Lemma \ref{lem:twoTypesCalculatedHighValueOrder} therefore gives
\[w_S(2,0,c)\geq w_S(2,1,c)\geq w_S(2,2,c)\geq w_S(1,2,c)>\kappa(y).\]
As both post-survival values in \eqref{eq:twoTypesSurvivalRuleHighLow} exceed $\kappa(y)$, $\alpha_S(2,1,y)=1$.
\end{proof}

\begin{proof}[Proof of Lemma \ref{lem:twoTypesCalculatedMonopolyValueOrder}]
The proof orders the entry rules, then the duopoly values, then the monopoly values.

First, comparing \eqref{eq:twoTypesBellmanLowHigh} and \eqref{eq:twoTypesBellmanLowLow} shows that the operator defined by the right-hand side of \eqref{eq:twoTypesBellmanLowLow} maps the set of functions $f$ satisfying $f(c)\geq w_S(1,2,c)$ for every $c\in\mathcal{C}$ into itself. Therefore, its fixed point satisfies $w_S(1,1,c)\geq w_S(1,2,c)$. This inequality, Assumption \ref{ass:entryBarriers}, and the construction of $\alpha_E(1,y)$ and $\alpha_E(2,y)$ therefore guarantee that $\alpha_E(1,y)\geq \alpha_E(2,y)$.

Next, we order the values of duopolists by their own technology types. Lemma \ref{lem:twoTypesCalculatedHighValueOrder} gives $w_S(2,0,c)\geq w_S(2,1,c)\geq w_S(2,2,c)\geq w_S(1,2,c)$. With these inequalities in hand, comparing \eqref{eq:twoTypesBellmanLowLow} and \eqref{eq:twoTypesBellmanHighLow} shows that the operator defined by the right-hand side of \eqref{eq:twoTypesBellmanLowLow} maps the set of functions $f$ satisfying $f(c)\leq w_S(2,1,c)$ for every $c\in\mathcal{C}$ into itself. Therefore, $w_S(1,1,c)\leq w_S(2,1,c)$.

Finally, the terms of \eqref{eq:twoTypesBellmanLowMonopoly} that follow the monopolist's advancement to type $2$ replicate the corresponding terms of \eqref{eq:twoTypesBellmanHighMonopoly}, because a type $1$ monopolist whose technology advances becomes a type $2$ monopolist. In the remaining terms, $\pi(1,0,y)\leq \pi(2,0,y)$, $\alpha_E(1,y)\geq\alpha_E(2,y)$, and $w_E(1,1,y)=\max\{0,w_S(1,1,c)-\kappa(y)\}\leq\max\{0,w_S(2,1,c)-\kappa(y)\}$. Together, these guarantee that the operator defined by the right-hand side of \eqref{eq:twoTypesBellmanLowMonopoly} maps the set of functions $f$ satisfying $f(c)\leq w_S(2,0,c)$ for every $c\in\mathcal{C}$ into itself. Therefore, $w_S(1,0,c)\leq w_S(2,0,c)$.
\end{proof}

\end{document}